\documentclass{statsoc}
\usepackage[english]{babel}
\usepackage{threeparttable}
\usepackage[utf8x]{inputenc}
\usepackage[dvipsnames,table,dvipsnames*, svgnames*, hyperref]{xcolor}

\newcommand{\ben}{\begin{enumerate}}
\newcommand{\een}{\end{enumerate}}
  \usepackage{siunitx}

\title[Transfer learning in High-dimensional Regression]{Transfer Learning for High-dimensional Linear Regression: Prediction, Estimation, and Minimax Optimality}

	\author[Sai Li]{Sai Li}
	\address{Department of Biostatistics, Epidemiology and Informatics, Perelman School of Medicine, University of Pennsylvania, Philadelphia, PA 19104 }
\author[T. Tony Cai]{T. Tony Cai}
	\address{Department of Statistics, The Wharton School, University of Pennsylvania, Philadelphia, PA 19104}
	\author[Li, Cai and Li]{Hongzhe Li}
 
		\address{Department of Biostatistics, Epidemiology and Informatics, Perelman School of Medicine, University of Pennsylvania, Philadelphia, PA 19104}

\usepackage{amssymb}
\usepackage{amsfonts, amsmath}
\usepackage{graphicx}
\usepackage{bbm}
\usepackage{multirow}
\usepackage{float}

\newcommand{\bel}{\begin{eqnarray}\label}
\newcommand{\eel}{\end{eqnarray}}
\newcommand{\bes}{\begin{eqnarray*}}
\newcommand{\ees}{\end{eqnarray*}}
\newcommand{\bei}{\begin{itemize}}
\newcommand{\beiftnt}{\begin{itemize}\footnotesize}
\newcommand{\eei}{\end{itemize}}

\def\benu{\begin{enumerate}}
\def\eenu{\end{enumerate}}

\def\argmin{\mathop{\rm arg\, min}}
\def\real{{\mathbb{R}}}
\def\R{{\real}}

\def\E{{\mathbb{E}}}
\def\P{{\mathbb{P}}}

\def\complex{\mathop{{\rm I}\kern-.58em\hbox{\rm C}}\nolimits}

\def\eps{\epsilon}

\def\lam{\lambda}

\usepackage{geometry}
\usepackage{natbib}
 \usepackage{hyperref}
\usepackage[titlenumbered,ruled]{algorithm2e}
\def\Sig{\Sigma}
\def\sig{\sigma}
\def\gam{\gamma}

\def\hT{\widehat{T}}

\def\mA{\mathcal{A}}
\def\hQ{\widehat{Q}}

 \newtheorem{theorem}{Theorem}
\newtheorem{lemma}{Lemma}

\newtheorem{corollary}{Corollary}

\newtheorem{remark}{Remark}
\newtheorem{condition}{Condition}
\newtheorem{theoremA}{Theorem}

\usepackage{xr}

\begin{document}
\maketitle
\begin{abstract}
This paper considers the estimation and prediction of  a high-dimensional linear regression in the setting of  transfer learning, using samples from the target model as well as auxiliary samples from different but possibly related regression models.  When the set of ``informative" auxiliary samples is known, an estimator and a predictor are proposed and their optimality is established. The optimal rates of convergence for prediction and estimation are faster than the corresponding rates without using the auxiliary samples. This implies that knowledge from the informative auxiliary samples can be transferred to improve the learning performance of the target problem. In the case that the set of informative auxiliary samples is unknown, we propose a data-driven procedure for transfer learning, called Trans-Lasso, and reveal its robustness to non-informative auxiliary samples and its efficiency in knowledge transfer.  The proposed procedures are demonstrated in numerical studies and are applied to a dataset concerning the associations among gene expressions. It is shown that Trans-Lasso leads to improved performance in gene expression prediction in a target tissue by incorporating the data from multiple different tissues as auxiliary samples. 
\end{abstract}

\section{Introduction}
\label{sec1}


Modern scientific research is characterized by a collection of massive and diverse data sets.  One of the most important goals is to integrate these different data sets for making better predictions and statistical inferences. Given a target problem to solve, transfer learning \citep{Torrey10} aims at transferring the knowledge from different but related samples to improve the learning performance of the target problem.  A typical example  of transfer learning is that one can improve the accuracy of recognizing cars by using not only the labeled data for cars  but some labeled data for trucks \citep{weiss2016survey}. 
Besides classification, another relevant class of transfer learning problems is linear regression using auxiliary samples. In health-related studies, some biological or clinical  outcomes are hard to obtain due to ethical or cost issues, in which case transfer learning  can be leveraged to boost the prediction and estimation performance of these outcomes by gathering information from different but related biological outcomes.

Transfer learning has been applied to problems in medical and biological applications, including   predictions of protein localization \citep{Mei11}, biological imaging diagnosis \citep{Shin16}, drug sensitivity prediction \citep{Turki17} and integrative analysis of``multi-omics'' data, see, for instance, \citet{Sun16}, \citet{Hu19}, and \citet{Wang19}.  It has also been applied to  natural language processing \citep{Daume07} and recommendation systems \citep{Pan13} in machine learning literature.  
The application that motivated the research in this paper is to integrate gene expression data sets measured in different issues to understand the gene regulations using the Genotype-Tissue Expression (GTEx) data (\url{https://gtexportal.org/}). 
 These datasets are always high-dimensional with relatively small sample sizes. When studying the gene regulation relationships  of a specific tissue or cell-type, it is possible to borrow information from other issues  in order to enhance the learning accuracy.  This motivates us to consider transfer learning in high-dimensional linear regression. 

\subsection{Transfer Learning in High-dimensional Linear Regression}
\label{sec1-1}
Regression analysis is one of  the most widely used statistical methods to understand the association of an outcome with a set of covariates. In many modern applications, the dimension of the covariates is usually very high as compared to the sample size. Typical examples include the genome-wide association and gene expression studies. 
In this paper, we consider transfer learning in high-dimensional linear regression models. Formally, our target model can be written as
\begin{equation}
\label{m0}
    y^{(0)}_i=(x_i^{(0)})^{\intercal}\beta+\eps^{(0)}_i,~i=1,\dots,n_0,
\end{equation}
where $((x_i^{(0)})^{\intercal}, y_i^{(0)}), ~i=1,\dots, n_0$, are independent samples, $\beta\in\R^p$ is the regression coefficient of interest, and $\eps_i^{(0)}$ are independently distributed random noises such that $\E[\eps_i^{(0)}|x_i^{(0)}]=0$. In the high-dimensional regime, where $p$ can be larger and much larger than $n_0$, $\beta$ is often assumed to be sparse such that the number of nonzero elements of $\beta$, denoted by $s$, is much smaller than $p$. 

In the context of transfer learning, we observe  additional samples from $K$ auxiliary studies,  That is, we observe $((x_i^{(k)})^{\intercal}, y_i^{(k)})$ generated from the auxiliary model 
\begin{align}
y_i^{(k)}=(x_i^{(k)})^{\intercal}w^{(k)}+\eps^{(k)}_i,~i=1,\dots,n_k,~k=1,\dots, K,
\end{align}
where $w^{(k)}\in\R^p$ is the true coefficient vector for the $k$-th study, and $\eps^{(k)}_i$ are the random noises such that $\E[\eps^{(k)}_i|x_i^{(k)}]=0$. 
 The regression coefficients $w^{(k)}$ are unknown and different from our target $\beta$ in general. The number of auxiliary studies, $K$, is allowed to grow but practically $K$ may not be too large. We will study the estimation and prediction of target model (\ref{m0}) utilizing the primary data $((x_i^{(0)})^{\intercal}, y_i^{(0)}), ~i=1,\dots, n_0,$ as well as the data from $K$ auxiliary studies $((x_i^{(k)})^{\intercal}, y_i^{(k)}),~ i=1,\dots, n_k,~k=1,\dots, K$.

If useful information can be borrowed from the auxiliary samples, the target model and some of the auxiliary models need to possess a certain level of similarity. If an auxiliary model is ``similar'' to the target model, we say that this auxiliary sample/study is informative. In this work, we characterize the informative level of the $k$-th auxiliary study  using the sparsity of the difference between $w^{(k)}$ and $\beta$. Let $\delta^{(k)}=\beta-w^{(k)}$ denote the contrast between $w^{(k)}$ and $\beta$. The set of informative auxiliary samples are those whose contrasts are sufficiently sparse:
\begin{equation}
\label{def-A}
\mA_q=\{1\leq k\leq K: \|\delta^{(k)}\|_q\leq h\}, 
\end{equation}
for some $q\in[0,1]$. That is, the set $\mA_q$, which contains the auxiliary studies  whose contrast vectors have $\ell_q$-sparsity at most $h$, is called the {\it informative set}. It will be seen later that as long as $h$ is relatively small to the sparsity of $\beta$, the studies in $\mA_q$ can be useful in improving the prediction and estimation of $\beta$. In the case of $q=0$, the set $\mA_q$ corresponds to the auxiliary samples whose contrast vectors have at most $h$ nonzero elements. We also consider approximate sparsity constraints $(q\in (0,1])$, which allows all of the coefficients to be nonzero but their absolute magnitude decays at a relatively rapid rate. 
For any $q\in[0,1]$, smaller $h$ implies that the auxiliary samples in $\mA_q$ are more informative;  larger cardinality of $\mA_q$ ($|\mA_q|$) implies that a larger number of informative auxiliary samples. Therefore, smaller $h$ and larger $|\mA_q|$ should be favorable. We allow $\mA_q$ to be empty in which case none of the auxiliary samples are informative. For the auxiliary samples outside of $\mA_q$, we do not assume sparse $\delta^{(k)}$ and hence $w^{(k)}$ can be very different from $\beta$ for $k\notin \mA_q$.

There is a paucity of methods and fundamental theoretical results for high-dimensional linear regression in the transfer learning setting. In the case where the set of informative auxiliary samples $\mA_q$ is known, there is a lack of rate optimal estimation and prediction methods. A closely related topic is multi-task learning \citep{Ando05, LMT09}, where the goal is to simultaneously estimate multiple models using multiple response data. The multi-task learning considered in \citet{LMT09} estimates multiple high-dimensional sparse linear models under the assumption that the support of all the regression coefficients are the same. The goal of transfer learning is however different, as one is only interested in estimating the target model and this remains to be a largely unsolved problem.   
 \citet{CW19} studied the  minimax and adaptive methods for nonparametric classification in the transfer learning setting under similarity assumptions on all the auxiliary samples to the target distribution \citep[Definition 5]{CW19}. 
In the more challenging setting where the  set $\mA_q$ is unknown as is typical in real applications, it is unclear how to avoid the effects of adversarial auxiliary samples.  Additional challenges include the heterogeneity among the design matrices, which does not arise in the conventional high-dimensional regression problems and hence requires novel proposals.

\subsection{Our Contributions}

In the setting where the informative set $\mA_q$ is known, we propose a transfer learning algorithm, called Oracle Trans-Lasso, for estimation and prediction of the target regression vector and prove its minimax optimality under mild conditions. The result demonstrates a faster rate of convergence when $\mA_q$ is non-empty and $h$ is sufficiently smaller than $s$, in which case the knowledge from the informative auxiliary samples can be optimally transferred to substantially help solve the regression problem under the target model. 

In the more challenging setting where $\mA_q$ is unknown a priori, we introduce a data-driven algorithm, called Trans-Lasso, to adapt to the unknown $\mA_q$. The adaption is achieved by aggregating a number of candidate estimators. The desirable properties of the aggregation method guarantee that the Trans-Lasso is not much worse than the best one among the candidate estimators. We carefully construct the candidate estimators and, leveraging the properties of aggregation, demonstrate the robustness and the efficiency of Trans-Lasso under mild conditions. In terms of robustness, the Trans-Lasso is guaranteed to be not much worse than the Lasso estimator using only the primary samples no matter how adversarial the auxiliary samples are. In terms of efficiency, the knowledge from a subset of the informative auxiliary samples can be transferred to the target problem under proper conditions.  Furthermore, If  the contrast vectors in the informative samples are sufficiently sparse,  the Trans-Lasso estimator performs as if the informative set $\mA_q$ is known.

When the distributions of the designs are distinct in different samples, the effects of heterogeneous designs are studied. The performance of the proposed algorithms is justified theoretically and numerically in various settings.

\subsection{Related Literature}
Methods for incorporating auxiliary information into statistical inference have received much recent interest. In this context, \cite{Tony19} and \cite{Xia2020GAP} studied  the two-sample larges-scale multiple testing problems. \cite{Ban19} considered the high-dimensional sparse estimation and \cite{Mao19} focused on matrix completion. 
The auxiliary information in the aforementioned papers is given as some extra covariates while we have some additional raw data, which are high-dimensional, and it is not trivial to find the best way to summarize the information. \citet{Bastani18} studied estimation and prediction in high-dimensional linear models with one informative auxiliary study, where the sample size of the auxiliary study is larger than the number of covariates. This work considers more general scenarios under weaker assumptions. Specifically, the sample size of auxiliary samples can be smaller than the number of covariates and some auxiliary studies can be non-informative, which is more practical in applications. 

The problem we study here is certainly related to the high-dimensional prediction and estimation in the conventional settings where only samples from the target model are available.
Several $\ell_1$ penalized or constrained minimization methods have been proposed for prediction and estimation for high-dimensional linear regression; see, for example, \citet{Lasso, FL01, Zou06, CT07, Zhang10}. The minimax optimal rates for estimation and prediction are studied in \citet{Raskutti11} and \citet{Ver12}. 

\subsection{Organization and Notation}
The rest of this paper is organized as follows. Section \ref{sec2} focuses on the setting where the informative set $\mA_q$ is known and with the sparsity in \eqref{def-A} measured in $\ell_1$-norm.  A transfer learning algorithm is proposed for estimation and prediction of the target regression vector and its minimax optimality is established. In Section \ref{sec3}, we study the estimation and prediction of the target model when $\mA_q$ is unknown for $q=1$. In Section \ref{sec4}, we justify the theoretical performance of our proposals under heterogeneous designs and extend our main algorithms to deal with $\ell_q$-sparse contrasts for $q\in[0,1)$. In Section \ref{sec-simu}, the numerical performance of the proposed methods is studied in various settings. In Section \ref{sec-data}, the proposed algorithms are applied to an analysis of a Genotype-Tissue Expression (GTEx) dataset to investigate  the association of gene expression of one gene with other genes in a target tissue by leveraging data measured on other related tissues or cell types. 

We finish this section with notation. Let $X^{(0)}\in\R^{n_0\times p}$ and $y^{(0)}\in\R^{n_0}$ denote the design matrix and the response vector for the primary data, respectively. Let $X^{(k)}\in\R^{n_k\times p}$ and $y^{(k)}\in\R^{n_k}$ denote the design matrix and the response vector for the $k$-th sample, respectively. For a class of matrices $R_l\in\R^{n_l\times p_0}$, $l\in \mathcal{L}$, we use $\{R_l\}_{l\in \mathcal{L}}$ to denote $R_l$, $l\in \mathcal{L}$. Let $n_{\mA_q}=\sum_{k\in \mA_q}n_k$.  
For a generic semi-positive definite matrix $\Sig\in\R^{m\times m}$, let $\Lambda_{\max}(\Sig)$ and $\Lambda_{\min}(\Sig)$ denote the largest and smallest eigenvalues of $\Sig$, respectively. Let $\textrm{Tr}(\Sig)$ denote the trace of $\Sig$. Let $e_j$ be such that its $j$-th element is 1 and all other elements are zero. Let $a\vee b$ denote $\max\{a,b\}$ and $a\wedge b$ denote $\min\{a,b\}$. We use $c,c_0,c_1,\dots$ to denote generic constants which can be different in different statements. Let $a_n=O(b_n)$ and $a_n\lesssim b_n$ denote $|a_n/b_n|\leq c<\infty$ for some constant $c$ when $n$ is large enough. Let $a_n\asymp b_n$ denote $a_n/b_n\rightarrow c$ for some positive constant $c$ as $n\rightarrow \infty$. Let $a_n=O_P(b_n)$ and $a_n\lesssim_{\P} b_n$ denote $\P(|a_n/b_n|\leq c)\rightarrow 1$ for some constant $c<\infty$. Let $a_n=o_P(b_n)$ denote $\P(|a_n/b_n|>c)\rightarrow 0$ for any constant $c>0$.
\section{Estimation with Known Informative Auxiliary Samples}
\label{sec2}

In this section, we consider transfer learning for high-dimensional linear regression when the  informative set $\mA_q$ is known. We focus on the $\ell_1$-sparse characterization of the contrast vectors and leave the $\ell_q$-sparsity, $q\in[0,1)$, to Section \ref{sec4}. The notation $\mA_1$ will be abbreviated as $\mA$ in the sequel without special emphasis.

\subsection{Oracle Trans-Lasso Algorithm}
\label{sec2-1}
We propose a transfer learning algorithm, called {\it Oracle Trans-Lasso},  for estimation and prediction when $\mA$ is known.  As an overview, we first compute an initial estimator using the primary sample and all the informative auxiliary samples. 
However, its probabilistic limit is biased from $\beta$ as $w^{(k)}\neq \beta$ in general. We then correct its bias using the primary data in the second step.  Algorithm \ref{algo1} formally presents our proposed Oracle Trans-Lasso algorithm. 


\vspace{0.1in}\begin{algorithm}
 \SetKwInOut{Input}{Input}
    \SetKwInOut{Output}{Output}
\SetAlgoLined
 \Input{Primary data $(X^{(0)},y^{(0)})$ and informative auxiliary samples $\{X^{(k)},y^{(k)}\}_{k\in\mA}$}
 \Output{$\hat{\beta}$}

\underline{Step 1}. Compute
\begin{align}
\label{eq-hbeta-init}
\hat{w}^{\mA}&=\argmin_{w\in\R^p} \Big\{\frac{1}{2(n_{\mA}+n_0)}\sum_{k \in \mA\cup\{0\}}\|y^{(k)}-X^{(k)}w\|_2^2+\lam_w\|w\|_1\Big\}\;
\end{align}
for $\lam_w=c_1\sqrt{\log p/(n_0+n_{\mA})}$ with some constant $c_1$.

\underline{Step 2}. Let 
\begin{equation}
\label{eq-hbeta}
\hat{\beta}=\hat{w}^{\mA}+\hat{\delta}^{\mA},
\end{equation}
 where
\begin{equation}
\label{eq-hgam}
  \hat{\delta}^{\mA}=\argmin_{\delta\in\R^p} \left\{\frac{1}{2n_0}\|y^{(0)}-X^{(0)}(\hat{w}^{\mA}+\delta)\|_2^2+\lam_{\delta}\|\delta\|_1\right\}
\end{equation}
for $\lam_{\delta}=c_2\sqrt{\log p/n_0}$ with some constant $c_2$.
 \caption{\textbf{Oracle Trans-Lasso algorithm}} \label{algo1}
\end{algorithm}

In Step 1, $\hat{w}^{\mA}$ is realized based on the Lasso \citep{Lasso} using the primary sample and all the informative auxiliary samples. Its probabilistic limit is $w^{\mA}$, which can be defined via the following moment condition
\[
   \E\left[\sum_{k \in \mA\cup\{0\}} (X^{(k)})^{\intercal}(y^{(k)}-X^{(k)}w^{\mA})\right]=0.
\]
Denoting $\E[x_i^{(k)}(x_i^{(k)})^{\intercal}]=\Sig^{(k)}$, $w^{\mA}$ has the following explicit form:
\begin{align}
\label{eq-wA}
w^{\mA}=\beta+\delta^{\mA}
\end{align}
for $\delta^{\mA}=\sum_{k \in \mA}\alpha_k\delta^{(k)}$ and $\alpha_k=n_k/(n_{\mA}+n_0)$, if $\Sig^{(k)}=\Sig^{(0)}$ for all $k\in \mA$. 
That is, the probabilistic limit of $\hat{w}^{\mA}$, $w^{\mA}$, has bias $\delta^{\mA}$, which is a weighted average of $\delta^{(k)}$. Step 1 is related to the approach for high-dimensional misspecified models \citep{Bu15} and moment estimators. The estimator $\hat{w}^{\mA}$ converges relatively fast as the sample size used in Step 1 is relatively large. Step 2 corrects the bias, $\delta^{\mA}$, using the primary samples. In fact, $\delta^{\mA}$ is a sparse high-dimensional vector whose $\ell_1$-norm is no larger than $h$. 
Hence, the error of step 2 is under control for a relatively small $h$. 
The choice of the tuning parameters $\lam_w$ and $\lam_{\delta}$ will be further specified in Theorem \ref{thm0-l1}.



\subsection{Theoretical Properties of Oracle Trans-Lasso}
Formally, the parameter space we consider can be written as
\begin{equation}
\label{eq-Thetaq}
  \Theta_{q}(s,h)=\left\{(\beta,\delta^{(1)},\dots,\delta^{(K)} ): \|\beta\|_0\leq s, ~\max_{k \in \mA_q} \|\delta^{(k)}\|_q\leq h\right\}
\end{equation}
for $\mA_q\subseteq \{1,\dots, K\}$ and $q\in[0,1]$. 
We study the rate of convergence for the Oracle Trans-Lasso algorithm under the following two conditions. 
\begin{condition}
\label{cond1}{\rm
For each $k\in \mA\cup\{0\}$, each row of $X^{(k)}$ is \textit{i.i.d.} Gaussian distributed with mean zero and covariance matrix $\Sig$. The smallest and largest eigenvalues of $\Sig$ are bounded away from zero and infinity, respectively. 
}\end{condition}

\begin{condition}
\label{cond2}{\rm
For each $k\in\mA\cup\{0\}$, the random noises $\eps^{(k)}_i$ are \textit{i.i.d.} sub-Gaussian distributed mean zero and variance $\sig^2_k$.   For some constant $C_0$, it holds that $\max_{\mA\cup\{0\} }\E[\exp\{t\eps^{(k)}_i\}]\leq \exp\{t^2C_0\}$ for all $t\in \R$ and $\max_{0\leq k\leq K} \E[(y_i^{(k)})^2]$ is bounded away from infinity. 
}\end{condition}

Condition \ref{cond1} assumes random designs with Gaussian distribution. The Gaussian assumption provides convenience for bounding the restricted eigenvalues of sample Gram matrices. Moreover, the designs are identically distributed for $k\in \mA\cup\{0\}$. This assumption is for simplifying some technical conditions and will be relaxed in Section \ref{sec4}. Without loss of generality, we also assume the design matrices are normalized such that $\|X^{(k)}_{.,j}\|_2^2=n_k$ and $\Sig_{j,j}=1$ for all $1\leq j\leq p$, $k\in \mA\cup\{0\}$. Condition \ref{cond2} assumes sub-Gaussian random noises for primary and informative auxiliary samples and the second moment of the response vector is finite. Conditions \ref{cond1} and \ref{cond2} put no assumptions on the non-informative auxiliary samples as they are not used  in the Oracle Trans-Lasso algorithm.
In the next theorem, we prove the convergence rate of the Oracle Trans-Lasso.

\begin{theorem}[Convergence Rate of Oracle Trans-Lasso]
\label{thm0-l1}
\text{  }\text{ }
Assume that Condition \ref{cond1} and Condition \ref{cond2} hold true. We take
$\lam_w=$\text{ } $\max_{k\in\mA\cup\{0\}}c_1\sqrt{\E[(y^{(k)}_i)^2]\log p/(n_{\mA}+n_0})$ and $\lam_\delta=c_2\sqrt{\log p/n_0}$ for some sufficiently large constants $c_1$ and $c_2$ only depending on $C_0$.
If $s\log p/(n_{\mA}+n_0)+h(\log p/n_0)^{1/2}=o((\log p/n_0)^{1/4})$, then it holds that
\begin{align}
&\sup_{\beta\in\Theta_1(s,h)}\frac{1}{n_0}\|X^{(0)}(\hat{\beta}-\beta)\|_2^2\vee\|\hat{\beta}-\beta\|_2^2\nonumber\\
&=O_P\left(\frac{s\log p}{n_{\mA}+n_0}+\frac{s\log p}{n_0}\wedge h\sqrt{\frac{\log p}{n_0}}\wedge h^2\right) \label{l1-re1}.
\end{align}
\end{theorem}
Theorem \ref{thm0-l1} provides the convergence rate of $\hat{\beta}$ for any $\beta\in\Theta_1(s,h)$.  
In the trivial case where $\mA$ is empty, the right-hand side in (\ref{l1-re1}) is $O_P(s\log p/n_0)$, which is the convergence rate for the Lasso only using primary samples. When $\mathcal{A}$ is non-empty, the right-hand side of (\ref{l1-re1}) is sharper than $s\log p/n_0$ if  $h\sqrt{\log p/n_0}\ll s$ and $n_{\mA}\gg n_0$. That is, if the informative auxiliary samples have contrast vectors sufficiently sparser than $\beta$ and the total sample size is significantly larger than the primary sample size, then the knowledge from the auxiliary samples can significantly improve the learning performance of the target model. In practice, even if $n_{\mA}$ is comparable to $n_0$, the Oracle Trans-Lasso can still improve the empirical performance as shown by some numerical experiments provided in Section \ref{sec-simu}.

The sample size requirement in Theorem \ref{thm0-l1} guarantees the lower restricted eigenvalues of the sample Gram matrices in Step 1 and Step 2 are bounded away from zero with high probability. The proof of Theorem \ref{thm0-l1} involves an error analysis of $\hat{w}^{\mA}$ and that of $\hat{\delta}^{\mA}$. While $w^{\mA}$ may be  neither $\ell_0$- nor $\ell_1$-sparse, it can be decomposed into an $\ell_0$-sparse component plus an $\ell_1$-sparse component as illustrated in (\ref{eq-wA}). Exploiting this sparse structure is a key step in proving Theorem \ref{thm0-l1}.
Regarding the choice of tuning parameters, $\lam_w$ depends on the second moment of $y_i^{(k)}$, which can be consistently estimated by $\|y^{(k)}\|_2^2/n_k$. The other tuning parameter $\lam_{\delta}$ depends on the noise levels, which can be estimated by the scaled Lasso \citep{scaled-lasso}.  In practice, cross validation can be performed for selecting tuning parameters.

 We now establish the minimax lower bound for estimating $\beta$ in the transfer learning setup, which shows the minimax optimality of the  Oracle Trans-Lasso algorithm in $\Theta_1(s,h)$.
\begin{theorem}[Minimax lower bound for $q=1$]
\label{thm2-low}
Assume Condition \ref{cond1} and Condition \ref{cond2}. If $\max\{s\log p/(n_{\mA}+n_0),~h(\log p/n_0)^{1/2}\}=o(1)$, then
{\small
\begin{align*}
&\P\left(\inf_{\hat{\beta}}\sup_{\beta\in \Theta_1(s,h)} \|\hat{\beta}-\beta\|_2^2\geq c_1\frac{s\log p}{n_{\mA}+n_0}+ c_2 \frac{s\log p}{n_0}\wedge h\left(\frac{\log p}{n_0}\right)^{1/2}\wedge h^2\right)\geq \frac{1}{2}
\end{align*}
}
for some positive constants $c_1$ and $c_2$.
\end{theorem}
Theorem \ref{thm2-low} implies that $\hat{\beta}$ obtained by the Oracle Trans-Lasso algorithm is minimax rate optimal in $\Theta_1(s,h)$ under the conditions of Theorem \ref{thm0-l1}. To understand the lower bound, the term $s\log p/(n_{\mA}+n_0)$ is the optimal convergence rate when $w^{(k)}=\beta$ for all $k \in \mA$. This is an extremely ideal case where we have $n_{\mA}+n_0$ \textit{i.i.d.} samples from the target model. The second term in the lower bound is the optimal convergence rate when $w^{(k)}=0$ for all $k \in \mA$, i.e., the auxiliary samples are not helpful at all. Let $\mathcal{B}_q(r)=\{u\in \R^p:\|u\|_q\leq r\}$ denote the $\ell_q$-ball with radius $r$ centered at zero. In this case, the definition of $\Theta_1(s,h)$ implies that $\beta\in \mathcal{B}_0(s)\cap \mathcal{B}_1(h)$ and the second term in the lower bound is indeed the minimax optimal rate for estimation when $\beta\in \mathcal{B}_0(s)\cap \mathcal{B}_1(h)$ with $n_0$ \textit{i.i.d.} samples \citep{TS14}.  

\section{Unknown Set of Informative Auxiliary Samples}
\label{sec3}
The  Oracle Trans-Lasso algorithm is based on the knowledge of the informative set $\mathcal{A}$. In some applications, the informative set $\mathcal{A}$ is not given, which makes the transfer learning problem more challenging. In this section, we propose a data-driven method for estimation and prediction when $\mathcal{A}$ is unknown. The proposed algorithm is described in detail in Section \ref{sec3-1} and \ref{sec3-2}. Its theoretical properties are studied in Section \ref{sec3-3}.

 
\subsection{The Trans-Lasso Algorithm}
\label{sec3-1}

Our proposed algorithm, called Trans-Lasso, consists of two main steps. First, we construct a collection of candidate estimators, where each of them is based on an estimate of $\mA$. Second, we perform an aggregation step \citep{RT11, DRZ12, Dai18} on these candidate estimators. Under proper conditions, the aggregated estimator is guaranteed to be not much worse than the best candidate estimator under consideration in terms of prediction. For technical reasons, we need the candidate estimators and the sample for aggregation to be independent. Hence, we start with sample splitting. 
We need some more notation. For a generic estimate of $\beta$, $b$, denote its sum of squared prediction error as 
\[
   \hQ(\mathcal{I},b)=\sum_{i\in \mathcal{I}}\|y^{(0)}_i-(x^{(0)}_i)^{\intercal}b\|_2^2,
\]
where $\mathcal{I}$ is a subset of $\{1,\dots,n_0\}$.
Let $\Lambda^{L+1}=\{\nu\in\R^{L+1}: \nu_l\geq 0,\sum_{l=0}^L\nu_l=1\}$ denote an $L$-dimensional simplex.
The Trans-Lasso algorithm is presented in Algorithm \ref{algo2}.

\begin{algorithm}
 \SetKwInOut{Input}{Input}
    \SetKwInOut{Output}{Output}
\SetAlgoLined
 \Input{ Primary data $(X^{(0)},y^{(0)})$ and samples from $K$ auxiliary  studies  $\{X^{(k)},y^{(k)}\}_{k=1}^K$.}
 \Output{$\hat{\beta}^{\hat{\theta}}$.}

\underline{Step 1}. 
Let $\mathcal{I}$ be a random subset of $\{1,\dots,n_0\}$ such that $|\mathcal{I}|\approx n_0/2$. Let $\mathcal{I}^c=\{1,\dots,n_0\}\setminus \mathcal{I}$. 

\underline{Step 2}.
 Construct $L+1$ candidate sets of $\mathcal{A}$, $\big\{\widehat{G}_0,\widehat{G}_1,\dots,\widehat{G}_{L}\big\}$ such that $\widehat{G}_0=\emptyset$ and $\widehat{G}_1,\dots,\widehat{G}_{L}$ are based on (\ref{eq-hGl}) using $\left(X_{\mathcal{I},.}^{(0)},y_{\mathcal{I}}^{(0)}\right)$ and $\{X^{(k)},y^{(k)}\}_{k=1}^M$.

\underline{Step 3}. For each $0\leq l\leq L$, run the Oracle Trans-Lasso algorithm with primary sample $(X_{\mathcal{I},.}^{(0)},y_{\mathcal{I}}^{(0)})$ and auxiliary samples $\{X^{(k)},y^{(k)}\}_{k\in\widehat{G}_l}$. Denote the output as $\hat{\beta}(\widehat{G}_l)$ for $0\leq l\leq L$.

\underline{Step 4}. 
Compute 
{\small
\begin{align}
  &\hat{\theta}=\label{theta-bma}\\
  &\argmin_{\theta\in\Lambda^{L+1}}\left\{\widehat{Q}\big(\mathcal{I}^c,\sum_{l=0}^L\hat{\beta}(\widehat{G}_l)\theta_l\big)+\sum_{l=0}^L\theta_l \widehat{Q}(\mathcal{I}^c,\hat{\beta}(\widehat{G}_l))+\frac{2\lam_{\theta}\log (L+1)}{n_0}\|\theta\|_1\right\}\nonumber
 \end{align}
 }
for some $\lam_{\theta}>0$. Output
\begin{equation}
\label{hbeta-htheta}
   \hat{\beta}^{\hat{\theta}}=\sum_{l=0}^L\hat{\theta}_{l}\hat{\beta}(\widehat{G}_l).
\end{equation}
 \caption{\textbf{Trans-Lasso Algorithm}} \label{algo2}
\end{algorithm}
As an illustration, steps 2 and 3 of the Trans-Lasso algorithm are devoted to constructing some initial estimates of $\beta$, $\hat{\beta}(\widehat{G}_l)$. They are computed using the Oracle Trans-Lasso algorithm by treating each $\widehat{G}_l$ as the set of informative auxiliary samples. We construct $\widehat{G}_l$ to be a class of estimates of $\mA$ and the detailed procedure is provided in Section \ref{sec3-2}.
Step 4 is based on the Q-aggregation proposed in \citet{DRZ12} with a uniform prior and a simplified tuning parameter. The Q-aggregation can be viewed as a weighted version of least square aggregation and exponential aggregation \citep{RT11} and it has been shown to be rate optimal both in expectation and with high probability for model selection aggregation problems.

The framework of model selection aggregation is a good fit for the transfer learning task under consideration. 
On one hand, it guarantees the robustness of Trans-Lasso in the following sense. Notice that $\hat{\beta}(\widehat{G}_0)$ corresponds to the Lasso estimator only using the primary samples and it is always included in our dictionary. The purpose is that, invoking the property of model selection aggregation, the performance of $\hat{\beta}^{\hat{\theta}}$ is guaranteed to be not much worse than the performance of the original Lasso estimator under mild conditions. This shows the performance of Trans-Lasso will not be ruined by adversarial auxiliary samples. Formal statements are provided in Section \ref{sec3-3}.
On the other hand, the gain of Trans-Lasso relates to the qualities of $\widehat{G}_1,\dots,\widehat{G}_L$.
If
\begin{equation}
\label{cond-agg}
\P\left(\widehat{G}_l\subseteq\mathcal{A}, ~\text{for some}~ 1\leq l\leq L\right)\rightarrow 1,
\end{equation} 
i.e., $\widehat{G}_l$ is a nonempty subset of the informative set $\mA$, then the model selection aggregation property implies that the performance of $\hat{\beta}^{\hat{\theta}}$ is not much worse than the performance of the Oracle Trans-Lasso with $\sum_{k\in \widehat{G}_l}n_k$ informative auxiliary samples. 
Ideally, one would like to achieve $\widehat{G}_l=\mathcal{A}$ for some $1\leq l\leq L$ with high probability. However, it can rely on strong assumptions that  may not be guaranteed in practical situations.

To motivate our constructions of $\widehat{G}_l$, let us first point out a naive construction of candidate sets, which consists of $2^K$ candidates. These candidates are all different combinations of $\{1,\dots, K\}$, denoted by $\widehat{G}_1,\dots, \widehat{G}_{2^K}$.
It is obvious that $\mA$ is an element of this candidate sets. However, the number of candidates is too large and it can be computationally burdensome. Furthermore, the cost of aggregation can be significantly high, which is of order $K/n_0$ as will be seen in Lemma \ref{thm-ag1}.
In contrast, we would like to pursue a much smaller number of candidate sets such that the cost of aggregation is almost negligible and (\ref{cond-agg}) can be achieved under mild conditions.
We introduce our proposed construction of candidate sets in the next subsection.

\subsection{Constructing the Candidate Sets for Aggregation}
\label{sec3-2}
As illustrated in Section \ref{sec3-1}, the goal of Step 2 is to have a class of candidate sets, $\{\widehat{G}_0,\dots,\widehat{G}_{L}\}$,  that satisfy (\ref{cond-agg}) under certain conditions. 
Our idea is to exploit the sparsity patterns of the contrast vectors.
Specifically, recall that the definition of $\mA$ implies that $\{\delta^{(k)}\}_{k \in \mA}$ are sparser than $\{\delta^{(k)}\}_{k \in \mA^c}$, where $\mA^c=\{1,\dots,K\}\setminus \mA$. This property motivates us to find a sparsity index $R^{(k)}$ and its estimator $\widehat{R}^{(k)}$ for each $1\leq k\leq K$ such that 
\begin{align}
\label{cond-rk}
 \max_{k \in \mA^o}R^{(k)}< \min_{k \in \mA^c}R^{(k)}\quad\text{and}\quad   \P\left(\max_{k \in \mA^o}\widehat{R}^{(k)}< \min_{k \in \mA^c}\widehat{R}^{(k)}\right)\rightarrow 1,
\end{align}
where $\mA^o$ is some subset of $\mA$.
In words, the sparsity indices in $\mA^o$ are no larger than the sparsity indices in $\mA^c$ and so are their estimators with high probability. To utilize (\ref{cond-rk}), we can define the candidate sets as
\begin{align}
\label{eq-hGl}
   \widehat{G}_l=\left\{1\leq k\leq K: \widehat{R}^{(k)}~ \text{is among the first $l$ smallest of all}\right\}
\end{align}
for $1\leq l\leq K$.
That is, $\widehat{G}_l$ is the set of auxiliary samples whose estimated sparsity indices are among the first $l$ smallest.
A direct consequence of (\ref{cond-rk}) and (\ref{eq-hGl}) is that $\P(\widehat{G}_{|\mA^o|}=\mA^o)\rightarrow 1$ and hence the desirable property (\ref{cond-agg}) is satisfied. 
To achieve the largest gain in transfer learning, we would like to find proper sparsity indices such that (\ref{cond-rk}) holds for $|\mA^o|$ as large as possible. 
Notice that $\widehat{G}_{K+1}=\{1,\dots, K\}$ is always included as candidates according to (\ref{eq-hGl}). Hence, in the special cases where all the auxiliary samples are informative or none of the auxiliary samples are informative, it holds that $\widehat{G}_{|\mA|}=\mA$ and the Trans-Lasso is not much worse than the Oracle Trans-Lasso.
The more challenging cases are $0<|\mA|<K$.

As $\{\delta^{(k)}\}_{k\in\mA^c}$ are not necessarily sparse, the estimation of $\delta^{(k)}$ or functions of $\delta^{(k)}$, $1\leq k\leq K$, is not trivial.
We consider using $R^{(k)}=\|\Sig\delta^{(k)}\|_2^2$, which is a function of the population-level marginal statistics, as the oracle sparsity index for $k$-th auxiliary sample. The advantage of $R^{(k)}$ is that it has a natural unbiased estimate without further assumptions. 
Let us relate $R^{(k)}$ to the sparsity of $\delta^{(k)}$ using a Bayesian characterization of sparse vectors assuming $\Sig^{(k)}=\Sig$ for all $0\leq k\leq K$. If $\delta^{(k)}_j$ are \textit{i.i.d.} Laplacian distributed with mean zero and variance $\nu_k^2$ for each $k$, then it follows from the properties of Laplacian distribution \citep{LK15} that
\[
   \E[\|\delta^{(k)}\|_1]=p\nu_k=\E^{1/2}[\|\Sig\delta^{(k)}\|^2_2]\frac{p}{\textrm{Tr}^{1/2}(\Sig\Sig)},
\]
where $\textrm{Tr}^{1/2}(\Sig\Sig)/p$ does not depend on $k$. Hence, the rank of $\E[\|\Sig\delta^{(k)}\|^2_2]$ is the same as the rank of $\E[\|\delta^{(k)}\|_1]$. As $\max_{k\in \mA}\|\delta^{(k)}\|_1< \min_{k\in \mA^c}\|\delta^{(k)}\|_1$, it is reasonable to expect $\max_{k\in \mA}\|\Sig\delta^{(k)}\|^2_2< \min_{k\in \mA^c}\|\Sig\delta^{(k)}\|^2_2$.
 Obviously, the above derivation holds for many other zero mean prior distributions besides Laplacian. This illustrates our motivation for considering $R^{(k)}$ as the oracle sparsity index.

We next  introduce the estimated version, $\widehat{R}^{(k)}$, based on the primary data $\{(x_i^{(0)})^{\intercal},y_{i}^{(0)}\}_{i\in \mathcal{I}}$ (after sample splitting) and auxiliary samples $\{X^{(k)},y^{(k)}\}_{k=1}^K$. 
We first perform a SURE screening \citep{FL08} on the marginal statistics to reduce the effects of  random noises.  
We summarize our proposal for Step 2 of the Trans-Lasso as follows (see Algorithm \ref{algo2.2}). Let $n_*=\min_{0\leq k\leq K}n_k$.

\begin{algorithm}[H]
\underline{Step 2.1}. 
For $1\leq k\leq K$, compute the marginal statistics
\begin{align}
\widehat{\Delta}^{(k)}&=\frac{1}{n_k}\sum_{i=1}^{n_k}x_i^{(k)}y_i^{(k)}-\frac{1}{|\mathcal{I}|}\sum_{i\in \mathcal{I}}x_i^{(0)}y_i^{(0)},\label{eq-hDeltak}\
\end{align}
For each $k\in\{1,\dots, K\}$, let $\hT_k$ be obtained by SURE screening such that 
\begin{align*}
  & \hT_k=\left\{1\leq j\leq p:~|\widehat{\Delta}_j^{(k)}|~\text{is among the first}~t_*~\text{largest of all}\right\}
   \end{align*}
for a fixed $t_*=n_*^{\alpha},~0\leq \alpha<1$.

\underline{Step 2.2}. Define the estimated sparse index for the $k$-th auxiliary sample as
\begin{align}
\label{eq-hRk}
   \widehat{R}^{(k)}=\left\|\widehat{\Delta}_{\hT_k}^{(k)}\right\|_2^2.
\end{align}
\underline{Step 2.3}. 
Compute $\widehat{G}_l$ as in (\ref{eq-hGl}) for $l=1,\dots, L$.
\caption{\textbf{Step 2 of the Trans-Lasso Algorithm}} \label{algo2.2}
\end{algorithm}

One can see that $\widehat{\Delta}^{(k)}$ are empirical marginal statistics such that $\E[\widehat{\Delta}^{(k)}]=\Sig\delta^{(k)}$ for $k\in\mA$.
The set $\hT_k$ is the set of first $t_*$ largest marginal statistics for the $k$-th sample. 
The purpose of screening the marginal statistics is to reduce the magnitude of noise. Notice that the un-screened version $\|\widehat{\Delta}^{(k)}\|_2^2$ is a sum of $p$ random variables and it contains noise of order $p/(n_k\wedge n_0)$, which diverges fast as $p$ is much larger than the sample sizes. By screening with $t_*$ of order $n_*^{\alpha}$, $\alpha<1$, the errors induced by the random noises is under control. In practice, the auxiliary samples with very small sample sizes can be removed from the analysis as their contributions to the target problem is mild.  
Desirable choices of $\hT_k$ should keep the variation of $\Sig\delta^{(k)}$ as much as possible.  
Under proper conditions, SURE screening can consistently select a set of strong marginal statistics and hence is appropriate for the current purpose. 
In Step 2.2, we compute $\widehat{R}^{(k)}$ based on the marginal statistics which are selected by SURE screening. 
In practice, different choices of $t_*$ may lead to different realizations of $\widehat{G}_l$. One can compute multiple sets of $\{\widehat{R}^{(k)}\}_{k=1}^K$ with different $t_*$ which give multiple sets of $\{\widehat{G}_l\}_{l=1}^K$. It will be seen from Lemma \ref{thm-ag1} that a finite number of choices on $t_*$ does not affect the rate of convergence.

\subsection{Theoretical Properties of Trans-Lasso}
\label{sec3-3}
In this subsection, we derive the theoretical guarantees for the Trans-Lasso algorithm. We first establish the model selection aggregation type of results for the Trans-Lasso estimator $\hat{\beta}^{\hat{\theta}}$.
\begin{lemma}[Q-aggregation for Trans-Lasso]
\label{thm-ag1}
Assume that Condition \ref{cond1} and Condition \ref{cond2} hold true. Let $\hat{\theta}$ be computed with 
$\lam_{\theta}\geq 4\sig_0^2$. With probability at least $1-t$, it holds that
\begin{align}
& \frac{1}{|\mathcal{I}^c|}\left\|X^{(0)}_{\mathcal{I}^c,.}(\hat{\beta}^{\hat{\theta}}-\beta)\right\|_2^2\leq \min_{0\leq l\leq L} \frac{1}{|\mathcal{I}^c|}\left\|X^{(0)}_{\mathcal{I}^c,.}(\hat{\beta}(\widehat{G}_l)-\beta)\right\|_2^2+\frac{\lam_{\theta}\log (L/t)}{n_0}.\label{re0-ms}
\end{align}
If $\|\Sig\|_2L\leq c_1n_0$ for some small enough constant $c_1$, then
\begin{align}
\label{re1-ms}
\left\|\hat{\beta}^{\hat{\theta}}-\beta\right\|_2^2\lesssim_{\P} \min_{0\leq l\leq L} \frac{1}{|\mathcal{I}^c|}\left\|X^{(0)}_{\mathcal{I}^c,.}(\hat{\beta}(\widehat{G}_l)-\beta)\right\|_2^2\vee \|\hat{\beta}(\widehat{G}_l)-\beta\|_2^2+\frac{\log L}{n_0}.
\end{align}
\end{lemma}
\begin{remark}
\label{re1}{\rm 
Assume that Conditions \ref{cond1} and \ref{cond2} hold. Let $\hat{\theta}$ be obtained with 
$\lam_{\theta}\geq 4\sig_0^2$. For any $L\geq 1$, it holds that
\[
 \big\|\hat{\beta}^{\hat{\theta}}-\beta\big\|_2^2\lesssim_{\P} \min_{0\leq l\leq L} \frac{1}{|\mathcal{I}^c|}\left\|X^{(0)}_{\mathcal{I}^c,.}(\hat{\beta}(\widehat{G}_l)-\beta)\right\|_2^2\vee \|\hat{\beta}(\widehat{G}_l)-\beta\|_2^2+\sqrt{\frac{\log L}{n_0}}.
 \]
}\end{remark}

Lemma \ref{thm-ag1} implies that the performance of $\hat{\beta}^{\hat{\theta}}$ only depends on the best candidate regardless of the performance of other candidates under mild conditions. 
As commented before, this result guarantees the robustness and efficiency of Trans-Lasso, which can be formally stated as follows. 
As the original Lasso is always in our dictionary, (\ref{re0-ms}) and (\ref{re1-ms}) imply that $\hat{\beta}^{\hat{\theta}}$ is not much worse than the Lasso in prediction and estimation. Formally, ``not much worse'' refers to the last term in (\ref{re0-ms}), which can be viewed as the cost of ``searching'' for the best candidate model within the dictionary which is of order $\log L/n_0$. This term is almost negligible, say, when $L=O(K)$, which corresponds to our constructed candidate estimators. This demonstrates the robustness of $\hat{\beta}^{\hat{\theta}}$ to adversarial auxiliary samples. Furthermore, if (\ref{cond-agg}) holds, then the prediction and estimation errors of Trans-Lasso are comparable to the Oracle Trans-Lasso based on auxiliary samples in $\mA^o$.

The prediction error bound in (\ref{re0-ms}) follows from Corollary 3.1 in \citet{DRZ12}. However, the aggregation methods do not have theoretical guarantees in estimation error in general. Indeed, an estimator with $\ell_2$-error guarantee is crucial for more challenging tasks, such as out-of-sample prediction and inference. For our transfer learning task, we show in (\ref{re1-ms}) that the estimation error is of the same order if the cardinality of the dictionary is $L\leq cn_0$ for some small enough $c$. For our constructed dictionary, it suffices to require $K\leq cn_0$. In many practical applications, $K$ is relatively small compared to the sample sizes and hence this assumption is not very strict.
In Remark \ref{re1}, we provide an upper bound on the estimation error which holds for arbitrarily large $L$ but is slower than (\ref{re1-ms}) in general.

In the following, we provide sufficient conditions such that the desirable property (\ref{cond-rk}) holds with $\widehat{R}^{(k)}$ defined in (\ref{eq-hRk}) and hence (\ref{cond-agg}) is satisfied. 
 For each $k \in\mA^c$, define a set
\begin{align}
\label{eq-Hk}
   H_k=\left\{1\leq j\leq p: |\Sig^{(k)}_{j,.}w^{(k)}-\Sig^{(0)}_{j,.}\beta|>n_*^{-\kappa}, ~\kappa<\alpha/2\right\}.
\end{align}
Recall that $\alpha$ is defined such that $t_*=n^\alpha$. In fact, $H_k$ is the set of ``strong'' marginal statistics that  can be consistently selected into $\widehat{T}_k$ for each $k\in \mA^c$.
We see that $\Sig^{(k)}_{j,.}w^{(k)}-\Sig^{(0)}_{j,.}\beta=\Sig_{j,.}\delta^{(k)}$ if $\Sig^{(k)}=\Sig^{(0)}$ for $k\in\mA^c$. The definition of $\mathcal{H}_k$ in (\ref{eq-Hk}) allows for heterogeneous designs among non-informative auxiliary samples.
\begin{condition}
\label{cond4}{\rm
(a) For each $k\in \mA^c$, each row of $X^{(k)}$ is \textit{i.i.d.} Gaussian with mean zero and covariance matrix $\Sig^{(k)}$. The largest eigenvalue of $\Sig^{(k)}$ is bounded away from infinity for any $k\in \mA^c$.
For each $k\in \mA^c$, the random noises $\eps^{(k)}_i$ are \textit{i.i.d.} Gaussian with mean zero and variance $\sig^2_k$. 

(b)It holds that $\log p\vee \log K\leq c_1 \sqrt{n_*}$ for a small enough constant $c_1$. Moreover,
\begin{align}
\label{eq1-cond4b}
 \min_{k \in\mA^c}\sum_{j\in H_k}|\Sig^{(k)}_{j,.}w^{(k)}-\Sig^{(0)}_{j,.}\beta|^2\geq \frac{c_1\log p}{n_*^{1-\alpha}}
\end{align}
for some large enough constant $c_1>0$.
} \end{condition}

The Gaussian assumptions in Condition \ref{cond4}(a) guarantee the desirable properties of SURE screening for the non-informative auxiliary studies. In fact, the Gaussian assumption can be relaxed to be sub-Gaussian random variables according to some recent studies \citep{SURE2}. For the conciseness of the proof, we consider Gaussian distributed random variables.
Condition \ref{cond4}(b) first puts constraint on the relative dimensions. It is trivial in the regime that $p\vee K\leq n_*^{\xi}$ for any finite $\xi>0$.  The expression (\ref{eq1-cond4b}) requires that for each $k\in\mA^c$, there exists a subset of strong marginal statistics such that their squared sum is beyond some noise barrier. This condition is mild by choosing $\alpha$ such that $\log p\ll n_*^{1-\alpha}$ and $\alpha=1/2$ is an obvious choice revoking the first part of Condition \ref{cond4}(b). For instance, if $\min_{k\in \mA^c}\|\E[\widehat{\Delta}^{(k)}]\|_{\infty}\geq c_0>0$, then (\ref{eq1-cond4b}) holds with any $\alpha\leq 1/2$. In words, a sufficient condition for (\ref{eq1-cond4b}) is that at least one marginal statistic in the $k$-th study is of constant order for $k\in\mA^c$. We see that larger $n_*$ makes Condition \ref{cond4} weaker. As mentioned before, it is helpful to remove the auxiliary samples with very small sample sizes from the analysis.

In the next theorem, we demonstrate the theoretical properties of $\widehat{R}^{(k)}$ and provide a complete analysis of the Trans-Lasso algorithm. Let $\mA^o$ be a subset of $\mA$ such that
\[
   \mA^o=\left\{k\in \mA:\|\Sig^{(0)}\delta^{(k)}\|_2^2\leq c_1\min_{k \in\mA^c}\sum_{j\in H_k}|\Sig^{(k)}_{j,.}w^{(k)}-\Sig^{(0)}_{j,.}\beta|^2\right\}
\]
for some $c_1<1$ and $H_k$ defined in (\ref{eq-Hk}). In general, one can see that the informative auxiliary samples with sparser $\delta^{(k)}$ are more likely to be included into $\mA^o$. Specially, the fact that $\max_{k\in \mA}\|\Sig^{(0)}\delta^{(k)}\|_2^2\leq \|\Sig^{(0)}\|_2^2h^2$ implies $\mA^o=\mA$ when $h$ is sufficiently small. We will show (\ref{cond-rk}) for such $\mA^o$ with $\widehat{R}^{(k)}$ defined in (\ref{eq-hRk}). 
 Let $n_{\mA^o}=\sum_{k\in \mA^o}n_k$.

\begin{theorem}[Convergence Rate of the Trans-Lasso]
\label{sec4-lem1}
Assume that the conditions of Theorem \ref{thm0-l1} and Condition \ref{cond4} hold. Then
\begin{align}
\label{cond-agg2}
   \P\left(\max_{k \in \mA^o}\widehat{R}^{(k)}< \min_{k \in \mA^c}\widehat{R}^{(k)} \right)\rightarrow 1.
\end{align}
Let $\hat{\beta}^{\hat{\theta}}$ be computed using the Trans-Lasso algorithm with $\lam_{\theta}\geq 4\sig^2_0$. If $K\leq cn_0$ for a sufficiently small constant $c>0$, then
\begin{align}
& \frac{1}{|\mathcal{I}^c|}\left\|X^{(0)}_{\mathcal{I}^c,.}(\hat{\beta}^{\hat{\theta}}-\beta)\right\|_2^2\vee\left\|\hat{\beta}^{\hat{\theta}}-\beta\right\|_2^2\nonumber\\
&=O_P\left(\frac{s\log p}{n_{\mA^o}+n_0}+\frac{s\log p}{n_0}\wedge h\sqrt{\frac{\log p}{n_0}}\wedge h^2+\frac{\log K}{n_0}\right).
\label{re3-agg}
\end{align}
\end{theorem}
\begin{remark}
\label{re3}{\rm
Under the conditions of Theorem \ref{sec4-lem1}, if
\[ \|\Sig^{(0)}\|_2^2h^2\leq \alpha\min_{k \in\mA^c}\sum_{j\in H_k}|\Sig^{(k)}_{j,.}w^{(k)}-\Sig^{(0)}_{j,.}\beta|^2~\text{for some}~\alpha<1,
\]
then $\P\left(\max_{k \in \mA}\widehat{R}^{(k)}< \min_{k \in \mA^c}\widehat{R}^{(k)} \right)\rightarrow 1$ and
\begin{align*}
& \frac{1}{|\mathcal{I}^c|}\left\|X^{(0)}_{\mathcal{I}^c,.}(\hat{\beta}^{\hat{\theta}}-\beta)\right\|_2^2\vee\left\|\hat{\beta}^{\hat{\theta}}-\beta\right\|_2^2\\
&=O_P\left(\frac{s\log p}{n_{\mA}+n_0}+\frac{s\log p}{n_0}\wedge h\sqrt{\frac{\log p}{n_0}}\wedge h^2+\frac{\log K}{n_0}\right).
\end{align*}
}\end{remark}
The result in (\ref{cond-agg2}) implies the estimated sparse indices in $\mA^o$ and in $\mA^c$ are separated with high probability. As illustrated before, a consequence of (\ref{cond-agg2}) is (\ref{cond-agg}) for the candidate sets $\widehat{G}_l$ defined in (\ref{eq-hGl}). Together with Theorem \ref{thm0-l1} and Lemma \ref{thm-ag1}, we arrive at (\ref{re3-agg}). 
In Remark \ref{re3}, we develop a sufficient condition for  $\mA^o=\mA$, which requires sufficiently small $h$. Under this condition, the estimation and prediction errors of the Trans-Lasso are comparable the  case where  $\mA$ is known, i.e. the adaptation to $\mA$ is achieved. Remark \ref{re3} is a direct consequence of Theorem \ref{sec4-lem1} and the fact that $\max_{k\in\mA}\|\Sig^{(0)}\delta^{(k)}\|_2^2\leq \|\Sig^{(0)}\|_2^2h^2$.

It is worth mentioning that Condition \ref{cond4} is only employed to show the gain of Trans-Lasso  and the robustness property of Trans-Lasso holds without any conditions on the non-informative samples (Lemma \ref{thm-ag1}). In practice, missing a few informative auxiliary samples may not be a very serious concern. One can see that when $n_{\mA^o}$ is large enough such that the first term on the right-hand side of (\ref{re3-agg}) no longer dominates, increasing the number of auxiliary samples will not improve the convergence rate. In contrast, it is more important to guarantee that the estimator is not affected by the adversarial auxiliary samples. 
The empirical performance of Trans-Lasso is carefully studied in Section \ref{sec-simu}.

\section{Extensions to Heterogeneous Designs and $\ell_q$-sparse Contrasts}
\label{sec4}
In this section, we extend the algorithms and theoretical results developed in Sections \ref{sec2} and \ref{sec3}. Section \ref{sec4-1} considers the case where the design matrices are heterogeneous with difference covariance structures and Section \ref{sec4-2} generalizes the sparse contrasts from $\ell_1$-constraint to  $\ell_q$-constraint for $q\in [0,1)$ and presents a rate-optimal estimator in this setting.

\subsection{Heterogeneous Designs}
\label{sec4-1}

The Oracle Trans-Lasso algorithm  proposed in Section \ref{sec2} can be directly applied to the  setting where the design matrices are heterogeneous.
To establish the theoretical guarantees in the heterogeneous case, we first introduce a relaxed version of Condition \ref{cond1} as follows.
\begin{condition}
\label{cond1b}{\rm
For each $k\in \mA\cup\{0\}$, each row of $X^{(k)}$ is \textit{i.i.d.} Gaussian with mean zero and covariance matrix $\Sig^{(k)}$. The smallest and largest eigenvalues of $\Sig^{(k)}$ are bounded away from zero and infinity, respectively, for all $k\in \mA\cup\{0\}$. 
}\end{condition}
Define
\[
    C_{\Sig}=1+\max_{j\leq p}\max_{k \in \mA}\Big\|e_j^{\intercal}\big(\Sig^{(k)}-\Sig^{(0)}\big)\Big(\sum_{k \in \mA\cup \{0\}}\alpha_k\Sig^{(k)}\Big)^{-1}\Big\|_1.
\]
The parameter $C_{\Sig}$ characterizes the differences between $\Sig^{(k)}$ and $\Sig^{(0)}$ for $k\in\mA$. Notice that $C_{\Sig}$ is a constant if $\max_{1\leq j\leq p}\|e_j^{\intercal}(\Sig^{(k)}-\Sig^{(0)})\|_0\leq C<\infty$ for all $k \in \mA$, where examples include block diagonal $\Sig^{(k)}$ with constant block sizes or banded $\Sig^{(k)}$ with constant bandwidths for $k \in \mA$.
The following theorem characterizes the rate of convergence of the Oracle Trans-Lasso estimator in terms of $C_{\Sig}$.

\begin{theorem}[Oracle Trans-Lasso with heterogeneous designs]
\label{thm1-l1}
Assume that Condition \ref{cond2} and Condition \ref{cond1b} hold true. We take $\lam_w= \max_{k\in\mA\cup\{0\}}c_1\sqrt{\E[(y^{(k)}_i)^2]\log p/(n_{\mA}+n_0})$ and $\lam_\delta=c_2\sig_0 \sqrt{\log p/n_0}$  for some sufficiently large constants $c_1$ and $c_2$ only depending on $C_0$.
If $s\log p/(n_{\mA}+n_0)+C_{\Sig}h(\log p/n_0)^{1/2}=o((\log p/n_0)^{1/4})$, then 
\begin{align}
&\frac{1}{n_0}\|X^{(0)}(\hat{\beta}-\beta)\|_2^2\vee\|\hat{\beta}-\beta\|_2^2\nonumber\\
&=O_P\left(\frac{s\log p}{n_{\mA}+n_0}+\frac{s\log p}{n_0}\wedge C_{\Sig}h\sqrt{\frac{\log p}{n_0}}\wedge C_{\Sig}^2h^2\right) \label{l1-re2}.
\end{align}
\end{theorem}
When $\mathcal{A}$ is non-empty, the right-hand side of (\ref{l1-re1}) is sharper than $s\log p/n_0$ if $n_{\mA_0}\gg n_0$ and $C_{\Sig}h\sqrt{\log p/n_0}\ll s$. We see that small $C_{\Sig}$ is favorable. This implies that  the informative auxiliary samples should not only have sparse contrasts but have similar Gram matrices to the primary one.
When $\mA$ is unknown, we consider $\tilde{\mA}^o$, a subset of $\mA$ such that
\[
  \tilde{\mA}^o=\left\{k\in \mA:\|\Sig^{(k)}w^{(k)}-\Sig^{(0)}\beta\|_2^2<c_1\min_{k \in\mA^c}\sum_{j\in H_k}|\Sig^{(k)}_{j,.}w^{(k)}-\Sig^{(0)}_{j,.}\beta|^2\right\}
\]
for some $c_1<1$ and $H_k$ defined in (\ref{eq-Hk}). This is a generalization of $\mA^o$ to the case of heterogeneous designs.

\begin{corollary}[Trans-Lasso with heterogeneous designs]
\label{cor1-l1}
Assume the conditions of Theorem \ref{thm1-l1} and Condition \ref{cond4}. 
Let $\hat{\beta}^{\hat{\theta}}$ be computed via the Trans-Lasso algorithm with $\lam_{\theta}\geq 4\sig^2_0$. If $K\leq cn_0$ for a small enough constant $c$, then
\begin{align}
& \frac{1}{|\mathcal{I}^c|}\left\|X^{(0)}_{\mathcal{I}^c,.}(\hat{\beta}^{\hat{\theta}}-\beta)\right\|_2^2\vee \|\hat{\beta}^{\hat{\theta}}-\beta\|_2^2\nonumber\\
&=O_P\left(\frac{s\log p}{n_{\tilde{\mA}^o}+n_0}+\frac{s\log p}{n_0}\wedge C_{\Sig}h\sqrt{\frac{\log p}{n_0}}\wedge C_{\Sig}^2h^2+\frac{\log K}{n_0}\right).
\label{re4-agg}
\end{align}
\end{corollary}
Corollary \ref{cor1-l1} provides an upper bound for the Tran-Lasso with heterogeneous designs. The numerical experiments for this setting are studied in Section \ref{sec-simu}. 

 \subsection{$\ell_q$-sparse Contrasts}
 \label{sec4-2}
 We have so far focused on the  $\ell_1$-sparse characterization of the contrast vectors. The established framework can be extended to the settings where the contrast vectors are characterized in terms of the $\ell_q$-norm for $q\in[0,1)$. We discuss the exact sparse contrasts ($q=0$) here and leave the results for $q\in(0,1)$  to the Appendix. We first consider the case when  $\mA_0$ is known and present the propopsed algorithm in Algorithm \ref{algo3}.

\vspace{0.1in}
\begin{algorithm}[H]
 \SetKwInOut{Input}{Input}
    \SetKwInOut{Output}{Output}
\SetAlgoLined
 \Input{Primary data $(X^{(0)},y^{(0)})$ and informative auxiliary samples $\{X^{(k)},y^{(k)}\}_{k\in\mA_0}$}
 \Output{$\hat{\beta}(\mA_0)$}

\underline{Step 1}. 
Estimate each $\delta^{(k)}$, $k\in\mA_0$, via
{\small
\begin{align*}
   &\hat{\delta}^{(k)}=\argmin_{\delta\in\R^p}\left\{\frac{1}{2}\delta^{\intercal}\widehat{\Sig}^{\mA_0}\delta-\delta^{\intercal}[(X^{(k)})^{\intercal}y^{(k)}/n_k-(X^{(0)})^{\intercal}y^{(0)}/n_0]+\lam_k\|\delta\|_1\right\},\nonumber
\end{align*}
}
where $\widehat{\Sig}^{\mA_0}=\sum_{k\in \mA_0\cup\{0\}}(X^{(k)})^{\intercal}X^{(k)}/(n_{\mA_0}+n_0)$ and $\lam_k>0$. 

\underline{Step 2}. Compute
{\small
\begin{align*}
&\hat{\beta}(\mA_0)=\argmin_{b\in\R^p} \Big\{\frac{1}{2(n_{\mA_0}+n_0)}\sum_{k \in \mA_0\cup\{0\}}\|y^{(k)}-X^{(k)}\hat{\delta}^{(k)}-X^{(k)}b\|_2^2+\lam_{\beta}\|b\|_1\Big\}\;\nonumber
\end{align*}
}
for some $\lam_{\beta}>0$.
 \caption{\textbf{Oracle Trans-Lasso algorithm for $q=0$}}\label{algo3}
\end{algorithm}

In the above algorithm, we estimate each $\delta^{(k)}$ based on the following moment equation:
\begin{align}
\label{eq-m1}
   \E[x^{(k)}_iy_i^{(k)}]- \E[x^{(0)}_iy_i^{(0)}]-\Sig\delta^{(k)}=0,~k\in \mA_0,
\end{align}
assuming that $\Sig^{(k)}=\Sig^{(0)}=\Sig$ for all $k \in \mA_0$. In the realization of Step 1, we replace $ \E[x^{(k)}_iy_i^{(k)}]$ and $\E[x^{(0)}_iy_i^{(0)}]$ by their unbiased sample versions, and the population Gram matrix $\Sig$ by its unbiased estimator $\widehat{\Sig}^{\mA_0}$. 
 In Step 2, we use the ``bias-corrected'' samples to estimate $\beta$. In contrast to the Oracle Trans-Lasso proposed in Section \ref{sec2-1}, the contrast vectors are estimated individually in the above algorithm. This is because the $\ell_1$-norm has the sub-additive property while $\ell_0$-norm does not. The computational cost in Step 1 can be relatively high if $\mA_0$ is large. Moreover, Step 1 heavily relies on the homogeneous designs among informative auxiliary samples.
In the next theorem, we prove the convergence rate of $\hat{\beta}(\mA_0)$. 
\begin{theorem}[Convergence rate of $\hat{\beta}(\mA_0)$]
\label{thm1-l0}
Assume that Condition \ref{cond1} and Condition \ref{cond2} hold true. Suppose that 
{\small
\begin{equation}
\label{ms-ss}
\frac{h\log p}{n_0}=o\Big(\big(\frac{\log p}{n_0+n_{\mA_0}}\big)^{1/4}\Big),~ \frac{s\log p}{n_{\mA_0}+n_0}=o(1),~\text{and}~~n_{\mA_0}\gtrsim |\mA_0|n_0.
\end{equation}
}
For $\hat{\beta}(\mA_0)$ computed with $\lam_k\geq c_1(\|y^{(k)}\|_2/n_k+ \|y^{(0)}\|_2/n_0)\sqrt{\log p}$ and \\ $\lam_{\beta}=c_2\sqrt{\log p/(n_0+n_{\mA_0})}$ for large enough constants $c_1$ and $c_2$ only depending on $C_0$,
{\small
\begin{align}
 &\sup_{\beta\in\Theta_0(s,h)}\frac{1}{n_0+n_{\mA_0}}\sum\limits_{k\in\mA_0\cup\{0\}}\|X^{(k)}(\hat{\beta}(\mA_0)-\beta)\|_2^2 \vee \|\hat{\beta}(\mA_0)-\beta\|_2^2\nonumber\\
 &=O_P\left( \frac{s\log p}{n_{\mA_0}+n_0}+\frac{(h\wedge s)\log p}{n_0}\right)\label{l0-re1}.
\end{align}
}
\end{theorem}
We see from (\ref{l0-re1}) that $\hat{\beta}(\mA_0)$ has sharper convergence rate than the Lasso when $n_{\mA_0}\gg n_0$ and $h\ll s$. The first two requirements in (\ref{ms-ss}) guarantee that the lower restricted eigenvalues of the sample Gram matrices are bounded away from zero. The last expression in (\ref{ms-ss}) requires the average auxiliary sample size is asymptotically no smaller than the primary sample size $n_0$. This is a checkable condition in practice. When $|\mA_0|$ is fixed, this condition can be trivially satisfied. Indeed, if the auxiliary sample sizes are too small, there would not be much improvement with transfer learning.  
We next establish the minimax lower bound for $\beta\in\Theta_0(s,h)$.
\begin{theorem}[Minimax lower bound for $q=0$]
\label{thm-mini2}
Assume Condition \ref{cond1} and Condition \ref{cond2}. Suppose that $\max\{h\log p/n_0, s\log p/(n_{\mA_0}+n_0)\}=o(1)$.
There exists some constants $c_1$ and $c_2$ such that
\begin{align*}
\P\left(\inf_{\hat{\beta}}\sup_{\Theta_0(s,h)} \|\hat{\beta}-\beta\|_2^2\geq c_1\frac{s\log p}{n_{\mA_0}+n_0}+ c_2\frac{(h\wedge s)\log p}{n_0}\right)\geq \frac{1}{2}.
\end{align*}
\end{theorem}
Theorem \ref{thm-mini2} demonstrates the minimax optimality of $\hat{\beta}(\mA_0)$ in $\Theta_0(s,h)$ under the conditions of Theorem \ref{thm1-l0}. 
When $\mA_0$ is unknown, one can consider the Trans-Lasso algorithm where the Oracle Trans-Lasso is replaced with the Oracle Trans-Lasso for $q=0$. We see that Lemma \ref{thm-ag1} still holds and a similar result as Theorem \ref{sec4-lem1} can be established with $\mA$ replaced by $\mA_0$. For the sake of conciseness, it is omitted here. 

While $\ell_0$-sparsity is widely assumed and studied in the high-dimensional literature, $\ell_0$-sparse contrast vectors may not be a realistic situation. First, an $\ell_0$-sparse $\delta^{(k)}$ implies that $\beta$ and $w^{(k)}$ have the same coefficients in most coordinates which may be impractical. Second, a typical data preprocessing step is to standardize the data such that $\|y^{(k)}\|_2^2=n_k$. While standardization does not change the $\ell_0$-norm of $\beta$ or $w^{(k)}$, it can change the $\ell_0$-norm of $\delta^{(k)}$. Hence, this work focuses on $\ell_1$-sparse contrasts, which is more practical in applications.



\section{Simulation Studies}
\label{sec-simu}
In this section, we evaluate the empirical performance of our proposals and some other comparable methods in various numerical experiments. Specifically, we evaluate the performance of four methods: the original Lasso, the Oracle Trans-Lasso proposed in Section \ref{sec2-1}, the Trans-Lasso proposed in Section \ref{sec3-1}, and a naive Trans-Lasso method, which naively assumes $\mA=\{1,\dots, K\}$ in the Oracle Trans-Lasso. The purpose of including the naive Trans-Lasso is to understand the overall informative level of the auxiliary samples. In the Appendix, we report the performance of the estimated sparse indices $\widehat{R}^{(k)}$ in achieving (\ref{cond-agg2}). 

\subsection{Identity Covariance Matrix for the Designs}
\label{sec5-1}
We consider $p=500$, $n_0=150$, and  $n_1,\dots,n_K=100$ for $K=20$. The covariates $x^{(k)}_i$ are \textit{i.i.d.} Gaussian with mean zero and identity covariance matrix for all $0\leq k\leq K$ and $\eps^{(k)}_i$ are \textit{i.i.d.} Gaussian with mean zero and variance one for all $0\leq k\leq K$. For the target parameter $\beta$, we set $s=16$, $\beta_{j}=0.3$ for $j\in \{1,\dots,s\}$, and $\beta_j=0$ otherwise. For the regression coefficients in auxiliary samples, we consider two configurations. 

(i) Let
\begin{align*}
   w^{(k)}_{j}=\beta_{j}-0.3\mathbbm{1}(j\in H_k).
   \end{align*}
For a given $\mA$, if $k \in \mA$, we set $H_k$ to be a random subset of $[p]$ with $|H_k|=h\in\{2,6,12\}$ and if $k\notin \mA$, we set $H_k$ to be a random subset of $[p]$ with $|H_k|=50$.

(ii)
For a given $\mathcal{A}$, if $k \in \mA$, let $H_k$ be a random subset of $[p]$ with $|H_k|=p/2$ and let 
\[
   w^{(k)}_{j}=\beta_{j}+\xi_j\mathbbm{1}(k\in H_k), ~\text{where}~ \xi_j\sim_{i.i.d.} \text{Laplace}(0, 2h/p),
\]
   where $h\in\{2,6,12\}$ and $\text{Laplace}(a, b)$ is Laplacian distribution with mean $a$ and dispersion $b$.
If $k\notin \mA$, we set $H_k$ to be a random subset of $[p]$ with $|H_k|=p/2$ and let
\[
    w^{(k)}_j=\beta_j+\xi_j\mathbbm{1}(j\in H_k),~\text{where}~\xi_j\sim_{i.i.d.} \text{Laplace}(0, 100/p).
\]
The setting (i) can be treated as either $\ell_0$- or $\ell_1$-sparse contrasts.
In practice, the true parameters are unknown and we use $\mA$ to denote the set of auxiliary samples without distinguishing $\ell_0$- or $\ell_1$-sparsity. We consider $|\mA|\in\{0,4,8,\dots,20\}$.

To avoid searching for tuning parameters, we use the raw data other than the standardized data.
 For the Lasso method, the tuning parameter is chosen to be $\sqrt{2\log p/n_0}$. For the Oracle Trans-Lasso, we set $\lam_w=\sqrt{2\log p/(n_0+n_{\mA})}$ and $\lam_{\delta}=\sqrt{2\log p/n_0}$. The naive Trans-Lasso is computed based on the Oracle Trans-Lasso with $\mA=\{1,\dots, K\}$. For the Trans-Lasso,  we set $\mathcal{I}=\{1,\dots,n_0/2\}$ in Step 1. The sets $\widehat{G}_0,\dots,\widehat{G}_L$ are computed based on SURE screening with $t_*=n_0^{3/4}$. For the Q-aggregation, we implement Algorithm 2 (GD-BMAX) in \citet{Dai18}, which solves a dual representation of the Q-aggregation. Using their notations, we set $\omega=\sqrt{2}$, $\nu=0.5$, and $t_k=1$ in GD-BMAX. We mention that as for demonstration, the tuning parameters are not deliberately chosen. We then perform a cross-fitting by using the first half of the samples for aggregation and the other half for constructing the dictionary. Our final estimator is the average of these two Trans-Lasso estimators.

The sum of squared estimation errors (SSE) are reported in Figure \ref{fig1-simu}. As we expected, the performance of the Lasso does not change as $|\mA|$ changes. On the other hand, all three other Trans-Lasso based algorithms have estimation errors decreasing as $|\mA|$ increases.  As $h$ increases, the problem gets harder and the estimation errors of all three methods increase. In settings (i) and (ii), the performance of Oracle Trans-Lasso and Trans-Lasso are comparable in most occasions. When $h=12$ in setting (i), we see a relatively large gap between the SSEs of Trans-Lasso and Oracle Trans-Lasso. One main reason is that the proposed $\widehat{R}^{(k)}$ does not consistently separate informative auxiliary samples from others in this case (Table \ref{table-rank} in the Appendix). This is because, the sparsity levels of $\delta^{(k)}$ are similar for $k\in\mA$ and $k\notin \mA$. In other cases, the proposed $\widehat{R}^{(k)}$ can separate informative auxiliary samples from others reasonably well.
On the other hand, the naive Trans-Lasso has worse performance than the Lasso when $|\mA|$ is relatively small. This shows that the scenarios under consideration are hard in the sense that some naive methods cannot adapt the unknown $\mA$ uniformly.
 
\begin{figure}
\makebox{\includegraphics[width=0.98\textwidth, height=4.5cm]{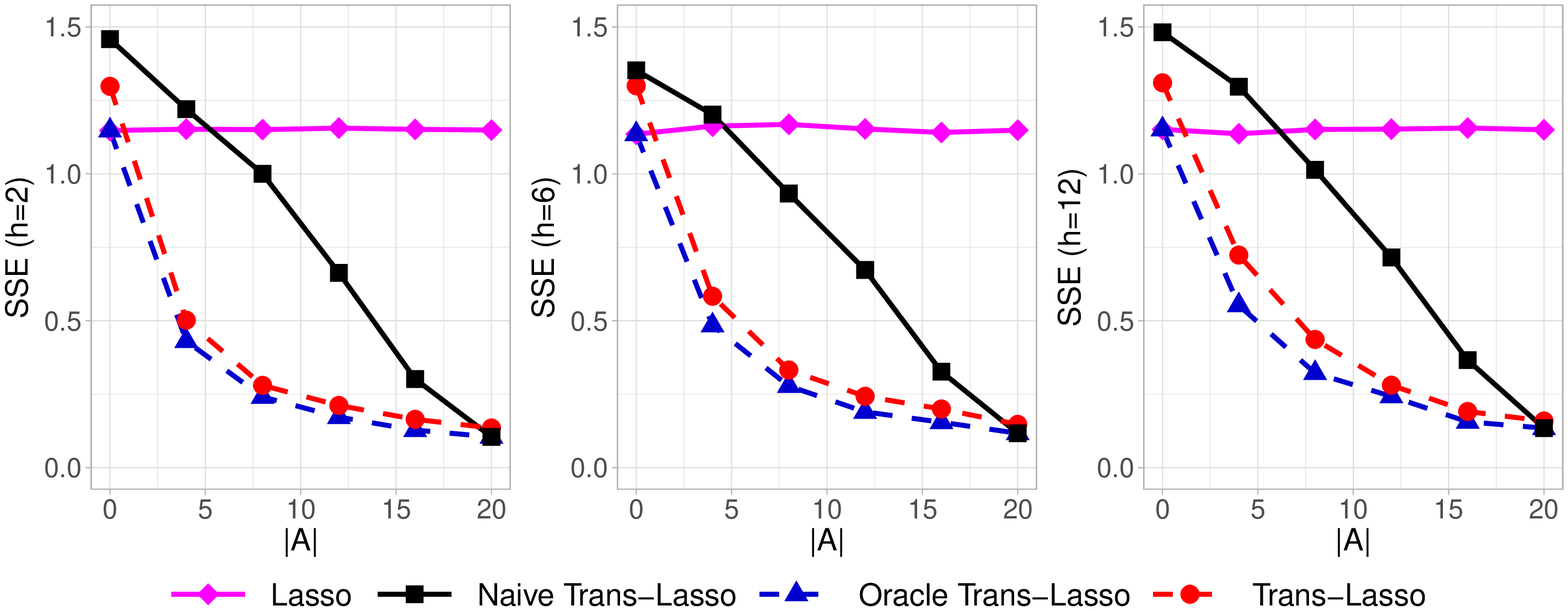}}
\makebox{\includegraphics[width=0.98\textwidth, height=4.5cm]{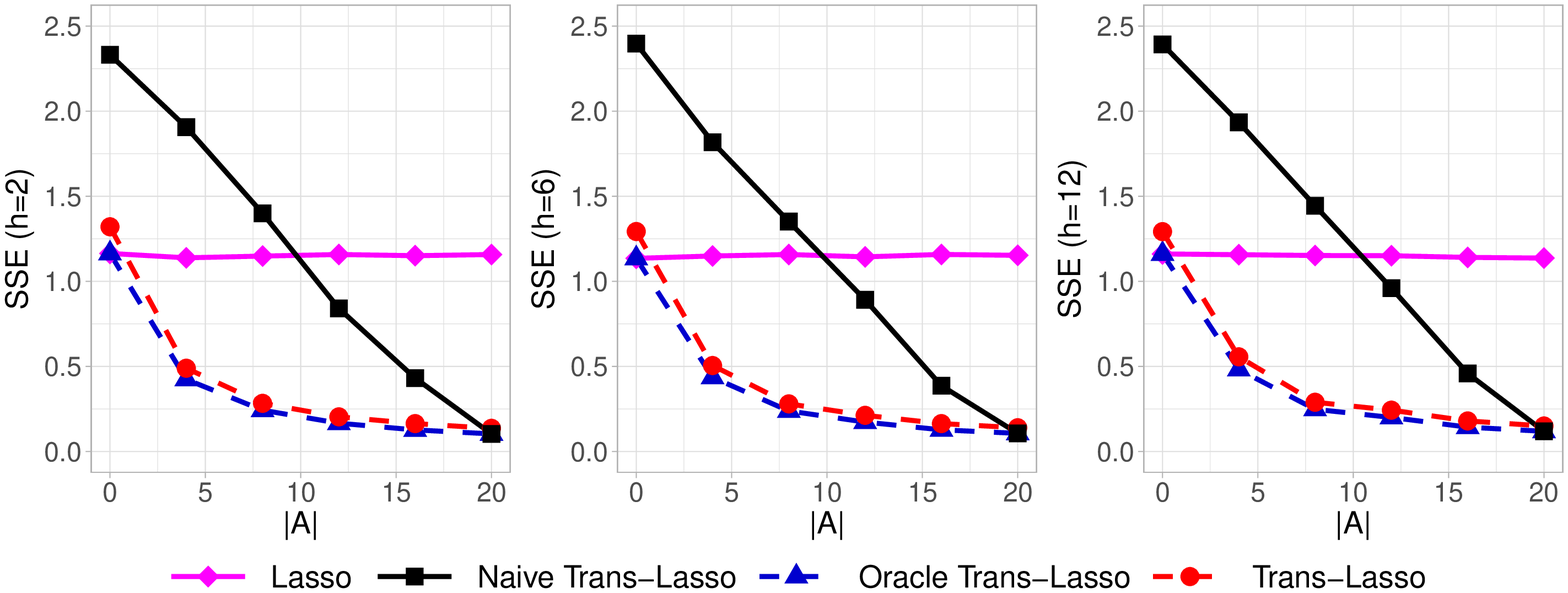}}
\caption{\label{fig1-simu}Estimation errors of the Lasso, Naive Trans-Lasso, Oracle Trans-Lasso, and Trans-Lasso for the settings with identity covariance matrices. The two rows correspond to configurations (i) and (ii), respectively. Each point is summarized from 200 independent simulations.}
\end{figure}

\subsection{Homogeneous Designs among $\mA\cup \{0\}$}
\label{sec5-2}
In this subsection, we consider $x^{(k)}_i$ as \textit{i.i.d.} Gaussian with mean zero and a Toeplitz covariance matrix whose first row is 
\[\Sig^{(k)}_{1,.}=(1,0.8,\dots,0.8^K, 0_{p-K-1})
\]
for $k\in \mA\cup\{0\}$. For $k\notin \mA\cup\{0\}$, $x^{(k)}_i$ are \textit{i.i.d.} Gaussian with mean zero and a Toeplitz covariance matrix whose first row is 
\begin{equation}
\label{eq-Sigk1}
\Sig^{(k)}_{1,.}=(1,\underbrace{1/(k+1),\dots,1/(k+1)}_{2k-1}, 0_{p-2k}).
\end{equation}
 Other true parameters and the dimensions of the samples are set to be the same as in Section \ref{sec5-1}. From the results presented in Figure \ref{fig2-simu}, we see that the Trans-Lasso and Oracle Trans-Lasso have reliable performance when $\Sig^{(k)}\neq \Sig^{(k')}$ for $k\in\mA\cup\{0\}$ and $k'\notin \mA\cup\{0\}$. In Table \ref{table-rank} in the Appendix, we see that our proposed $\widehat{R}^{(k)}$ can separate informative auxiliary samples from others consistently in settings (i) and (ii).
 
 On the other hand, we can observe from Figure \ref{fig2-simu} that the SSE of the Trans-Lasso can be slightly below those of the Oracle Trans-Lasso when $0<|\mA|< K$. There are two potential reasons. As a cross-fitting step is performed in the Trans-Lasso, the samples for computing the Trans-Lasso and for other methods are different empirically. Second, our definition of $\mA$ may not always be the best subset of auxiliary samples that  give the smallest estimation errors. 
\begin{figure}
\makebox{\includegraphics[width=0.98\textwidth, height=4.5cm]{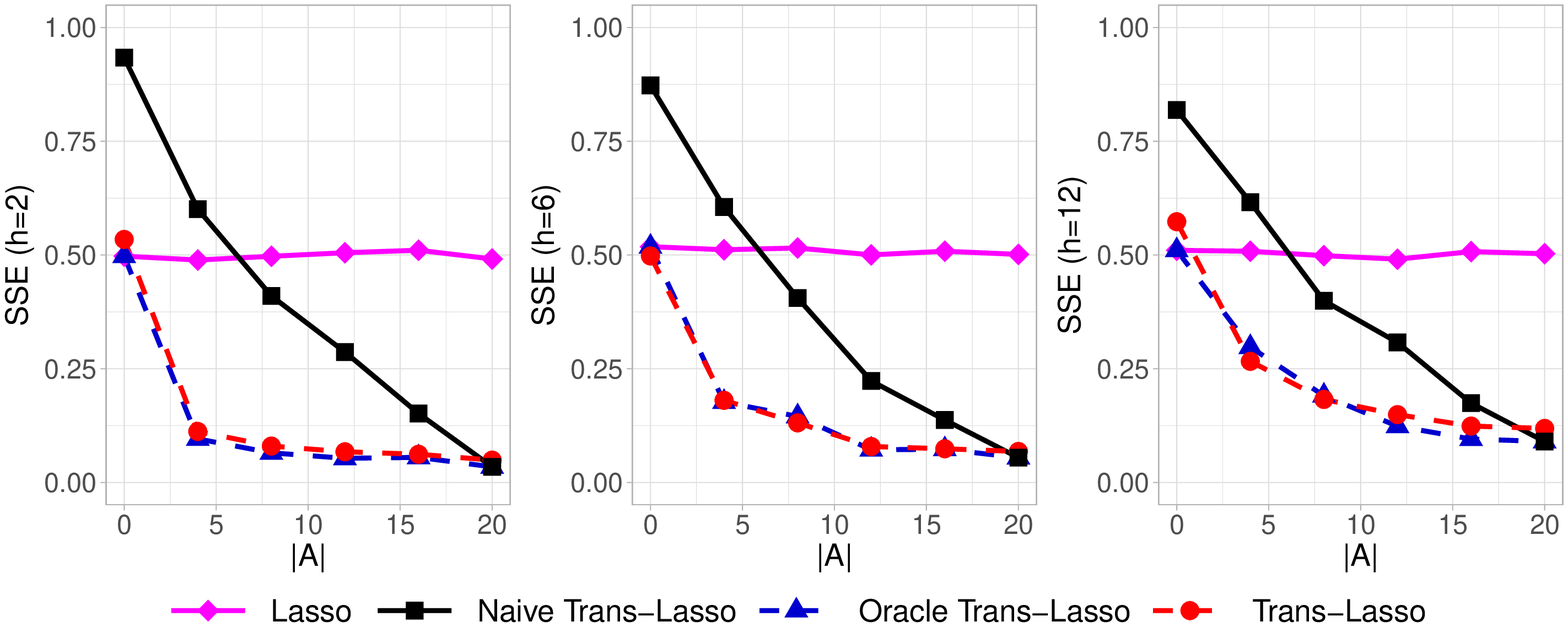}}
\makebox{\includegraphics[width=0.98\textwidth, height=4.5cm]{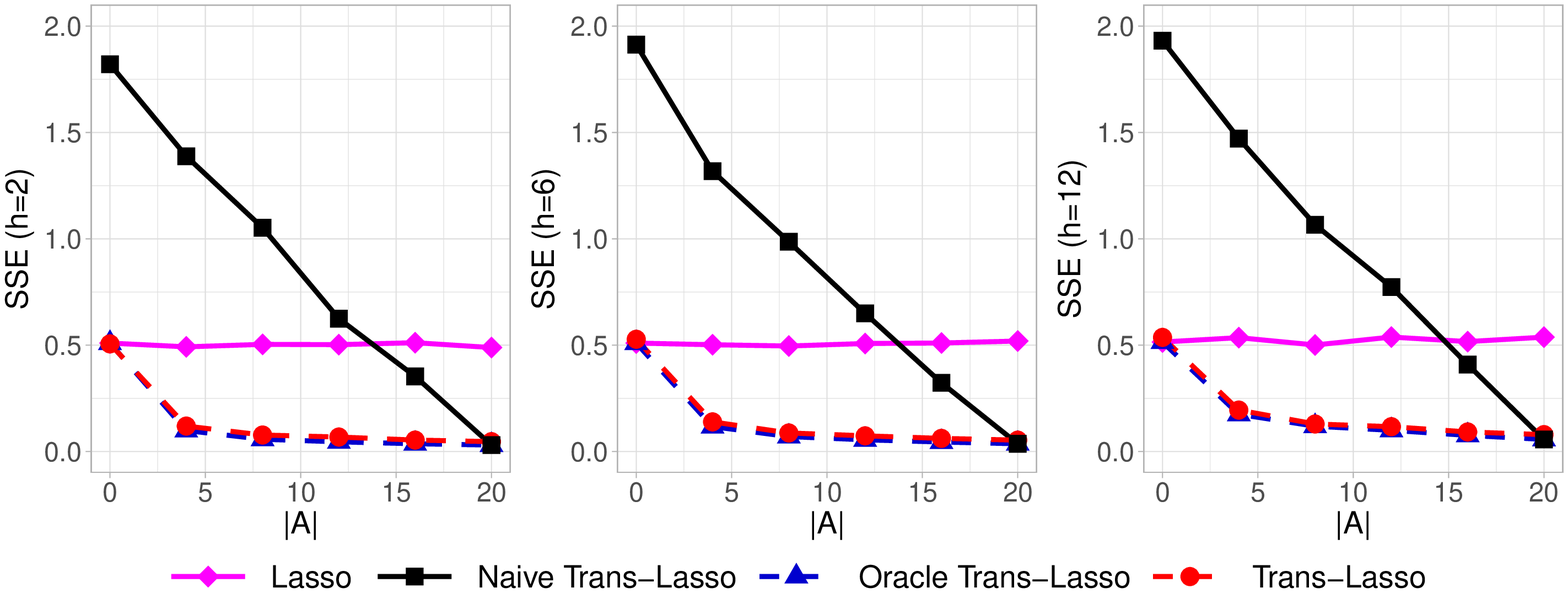}}
\caption{\label{fig2-simu}Estimation errors of the Lasso, Naive Trans-Lasso, Oracle Trans-Lasso, and Trans-Lasso for the settings with homogeneous covariance matrices among $k\in\mA\cup\{0\}$. The two rows correspond to configurations (i) and (ii), respectively. Each point is summarized from 200 independent simulations.}
\end{figure}

\subsection{Heterogeneous Designs}
We now consider $x^{(k)}_i$ as \textit{i.i.d.} Gaussian with mean zero and a Toeplitz covariance matrix whose first row is (\ref{eq-Sigk1}) for $k=1,\dots, K$. Moreover, $\Sig^{(0)}=I_p$. Other parameters and the dimensions of the samples are set to be the same as in Section \ref{sec5-1}.
Figure \ref{simu-fig4}  shows that the general patterns observed in previous subsections still hold. We observe a relatively large gap between the SSEs of the Oracle Trans-Lasso and Trans-Lasso in the scenario when $h=12$ in setting (i). Again, this is because $\widehat{R}^{(k)}$ can only separate a subset of $\mA$ from $\mA^c$. In other cases, the performance of Trans-Lasso is comparable to the Oracle Trans-Lasso.
\begin{figure}
\makebox{\includegraphics[width=0.98\textwidth, height=4.5cm]{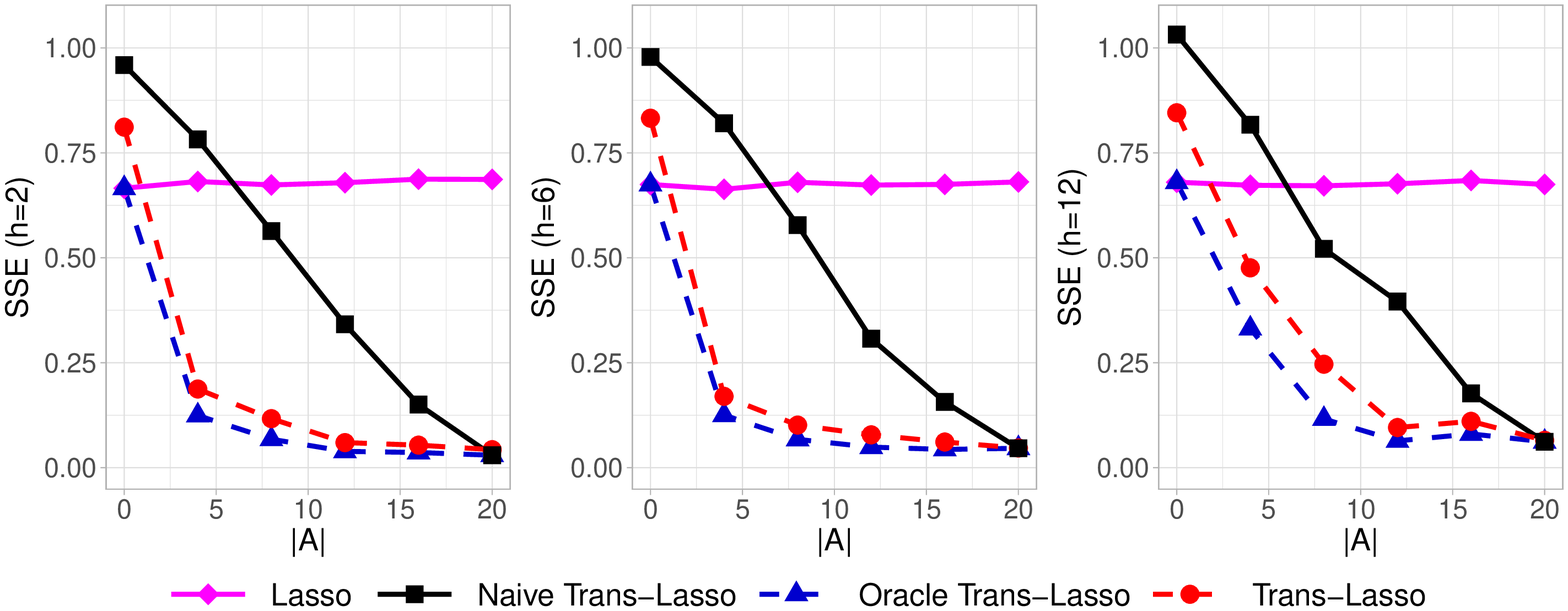}}
\makebox{\includegraphics[width=0.98\textwidth, height=4.5cm]{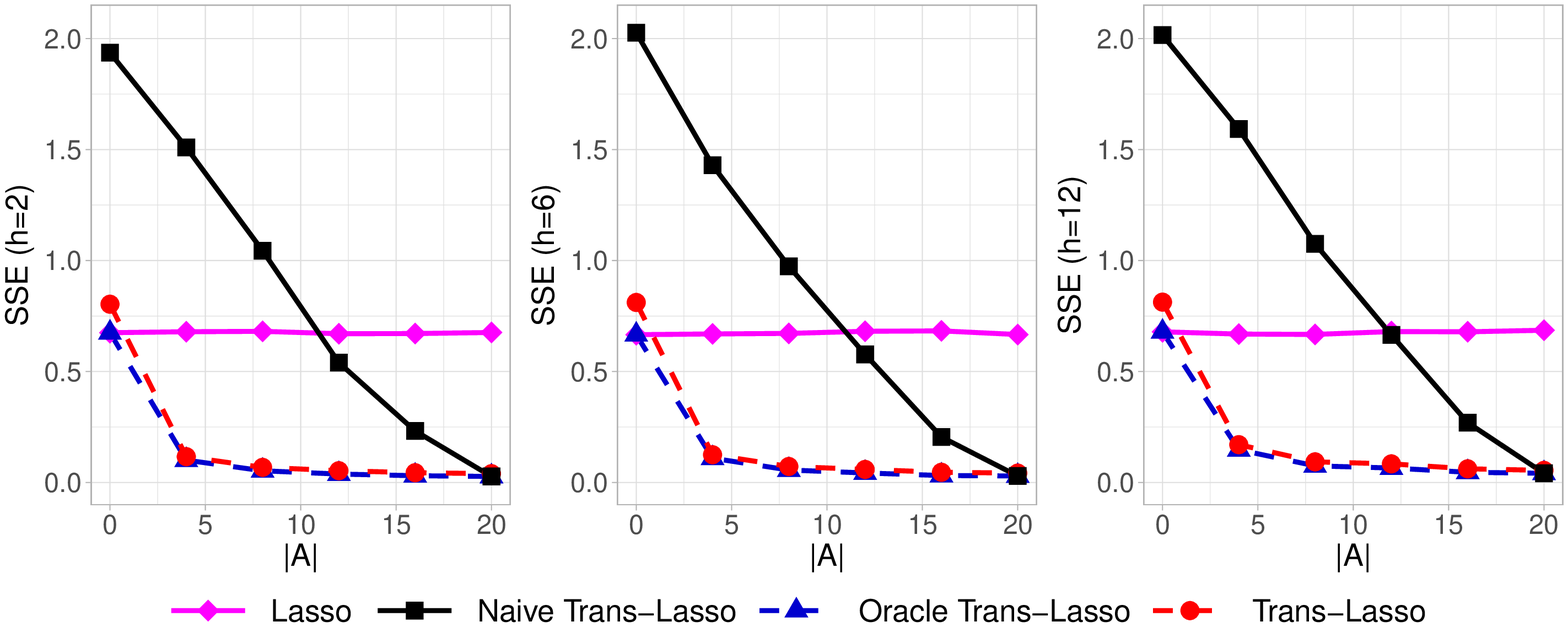}}
\caption{\label{simu-fig4}Estimation errors of the Lasso, Naive Trans-Lasso, Oracle Trans-Lasso, and Trans-Lasso for the settings with heterogeneous covariance matrices. The two rows correspond to configurations (i) and (ii), respectively. Each point is summarized from 200 independent simulations.}
\end{figure}

\section{Application to Genotype-Tissue Expression Data}
\label{sec-data}
In this section, we demonstrate  the performance of our proposed transfer learning algorithm in analyzing the Genotype-Tissue Expression (GTEx) data (\url{https://gtexportal.org/}). Overall, the data sets  measure  gene expression levels from 49 tissues of 838 human donors, in total comprising 1,207,976 observations of 38,187 genes.  In our analysis, we focus on genes that are related to central nervous systems, which were assembled as MODULE\underline{ }137 ( \url{https://www.gsea-msigdb.org/gsea/msigdb/cards/MODULE_137.html}). 
This module includes a total of  545 genes  and additional 1,632 genes that  are significantly enriched in the same experiments as the genes of the module.  A   complete list of genes can be found at \url{http://robotics.stanford.edu/~erans/cancer/modules/module_137}. 

\subsection{Data Analysis Method}
To demonstrate the replicability of our proposal, we consider multiple target genes and multiple target tissues and estimate their corresponding models one by one.
For an illustration of the computation process,  we  consider gene \texttt{JAM2} (Junctional adhesion molecule B), as the response variable, and treat other genes in this module as covariates.  \texttt{JAM2} is a protein coding gene on chromosome 21 interacting with a variety of immune cell types and may play a role in lymphocyte homing to secondary lymphoid organs \citep{JAM2}.  It is of biological interest to understand how other CNS genes  can predict its expression levels in different tissues/cell types. 

As an example, we consider the association between \texttt{JAM2} and other genes in a brain tissue as the target models and the association between \texttt{JAM2} and other genes in other tissues as the auxiliary models. As there are multiple brain tissues in the dataset, we treat each of them as the target at each time. The list of target tissues can be found in Figure \ref{fig-data-1}. The min, average, and max of primary sample sizes in these target tissues are 126, 177, and 237, respectively. The gene \texttt{JAM2} is expressed in 49 tissues in our dataset and we use 47 tissues  with more than 120 measurements on \texttt{JAM2}. The average number of auxiliary samples for each target model is 14,837 over all the non-target tissues. The covariates used are the  genes that are in the enriched MODULE\underline{ }137 and do not have missing values in all of the 47 tissues. The final covariates include a total of  1,079 genes. The data is standardized before analysis.

 We compare the prediction performance of Trans-Lasso with the Lasso. To understand the overall informative level of the auxiliary samples, we also compute the Naive Trans-Lasso which treats all the auxiliary samples are informative. For evaluation, we split the target sample into 5 folds and use 4 folds to train the three algorithms and use the remaining fold to test their prediction performance. We repeat this process 5 times each with a different fold of test samples. 
We mention that one individual can provide expression measurements on multiple tissues and these measurements are hard to be independent. As the dependence of the samples can reduce the efficiency of the estimation algorithms, using auxiliary samples may still be beneficial. However, one need to choose proper tuning parameters. The tuning parameter for the Lasso and $\lam_w$ in the Naive-Trans-Lasso are chosen by 8-fold cross validation. The $\lam_{\delta}$ in the Naive Trans-Lasso is set to be $\lam_w\sqrt{\sum_{k=0}^K n_k/n_0}$. For our proposal Trans-Lasso, we use two-thirds of the training sample to construct the dictionary and use one-third of the training sample for aggregation. The sparsity indices are computed in the same way as in our simulations. For computing each $\hat{\beta}(\widehat{G}_l)$, the $\lam_w$ is chosen by 8-fold cross validation and $\lam_{\delta}$ is set to be the corresponding $\lam_w\sqrt{\sum_{k\in \widehat{G}_l} n_k/n_0}$. The tuning parameters in aggregation are chosen as in our simulation.

\subsection{Prediction Performance of the Trans-Lasso for \texttt{JAM2} Expression}
Figure \ref{fig-data-1} demonstrates the errors of Naive Trans-Lasso, and Trans-Lasso relative to the Lasso for predicting gene expression \texttt{JAM2} using other genes. The prediction errors in the raw scale are provided in the Appendix. We see that the Trans-Lasso algorithm always achieves the smallest prediction errors across different tissues. Its average gain is 17\% comparing to the Lasso. This shows that our characterization of the similarity of the target model and a given auxiliary model is suitable to the current problem and our proposed sparsity index for aggregation is effective in detecting good auxiliary samples. In tissues such as Amygdala and Nucleus accumbens basal ganglia, the Trans-Lasso achieves relatively significant improvement and it has more accurate prediction than the Naive Trans-Lasso. This implies that the knowledge from the auxiliary tissues have been successfully transferred into these target tissues for modeling \texttt{JAM2} even if not all the tissues are informative. In tissues such as Pituitary, the improvement of the Naive Trans-Lasso and Trans-Lasso are mild. This implies that the regression model for \texttt{JAM2} in Pituitary is relatively distinct from the models in other tissues so that little knowledge can be transferred. Moreover, in tissues such as Frontal Cortex, the prediction performance of naive Trans-Lasso can be worse than the Lasso whereas the Trans-Lasso still has the smallest prediction error. This again demonstrates the robustness of our proposal.

 \begin{figure}[H]
 \centering
 \includegraphics[height=6cm, width=0.95\textwidth]{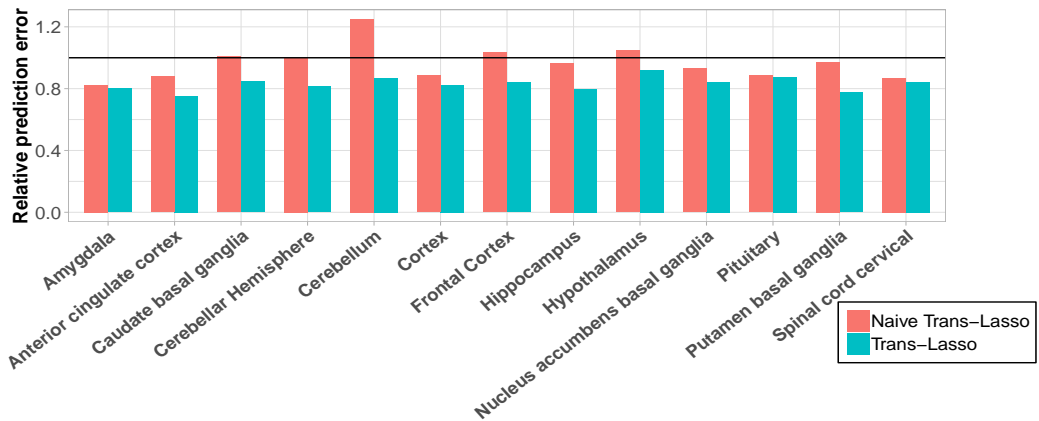}
 \caption{ \label{fig-data-1}Prediction errors of the naive Trans-Lasso and Trans-Lasso relative to the Lasso evaluated via 5-fold cross validation for gene \texttt{JAM2} in multiple tissues.}
 \end{figure}
 
 \subsection{Prediction Performance of Other 25 Genes on Chromsome 21}
To demonstrate the replicability of our proposal, we also consider other genes on Chromosome 21 which are in Module\underline{ }137 as our target genes. We report the overall prediction performance of these 25 genes in Figure \ref{fig-data-2}. A complete list of these genes and some summary information can be found in the Appendix. Generally speaking, we see that the Trans-Lasso has the best overall performance among all target tissues. Specifically, the improvement of Trans-Lasso is significant in tissues including Cerebellar Hemisphere, Cortex, and Frontal Cortex. The naive Trans-Lasso is has comparable accuracy to the Lasso in most cases, which implies that the overall similarity between the auxiliary tissues and target tissues is not very strong.

 \begin{figure}
 \centering
 \makebox{\includegraphics[height=6cm, width=0.99\textwidth]{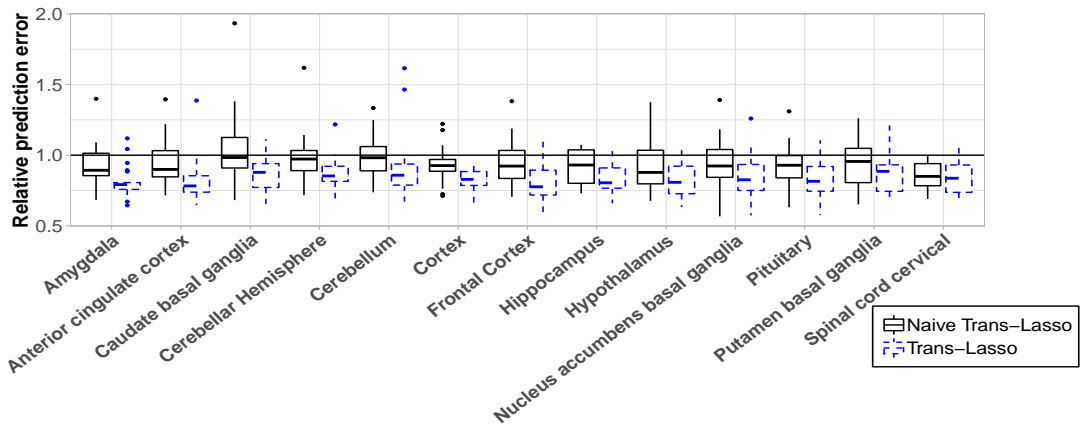}}
 \caption{ \label{fig-data-2}Prediction errors of the naive Trans-Lasso and Trans-Lasso relative to the Lasso for the 25 genes on Chromosome 21 and in Module\underline{ }137, in multiple target tissues.}

 \end{figure} 

\section{Discussion}
\label{sec:discussion}

This paper studies high-dimensional linear regression in the presence of additional auxiliary samples, where the similarity of the target model and a given auxiliary model is characterized by  the sparsity of their contrast vector.  Transfer learning algorithms for estimation and prediction are developed. The results show that if the informative set is known, the proposed Oracle Trans-Lasso is minimax optimal over a range of parameter spaces and the accuracy for estimation and prediction can be improved. 
It has been considered in \citet{Bastani18} the setting of known informative set with one large-scale auxiliary study, which is a special case of our problem set-up. However, their upper bound analysis is not minimax rate optimal in the parameter space considered in this work.
Adaptation to the unknown informative set is also considered. It is shown that adaptation can be achieved by aggregating a collection of candidate estimators. Numerical experiments and real data applications support the theoretical findings.

Transfer learning for high-dimensional linear regression  is an important problem with a wide range of potential applications. However, statistical theory and methods have not been well developed in the literature. Using our similarity characterization of the auxiliary studies, it is also interesting to study statistical inference such as constructing confidence intervals and hypothesis testing for high-dimensional linear regression with auxiliary samples. 
In view of the results derived in this paper, one may expect weaker sample size conditions in the transfer learning setting than the conventional case. It is interesting to provide a precise characterization and develop a minimax optimal confidence interval in the transfer learning setting. On the other hand, different measurements of the similarity can be used when they are appropriate, which can lead to new methods and insights into the underlying structure of the transfer learning algorithm. 

\section*{Acknowledgements}
This research was supported by NIH grants GM129781 and GM123056 and NSF Grant DMS-1712735.

\bibliography{Aggregation.bib}{}

\begin{thebibliography}{}

\bibitem[\protect\citeauthoryear{Ahmed and Bajwa}{Ahmed and
  Bajwa}{2019}]{SURE2}
Ahmed, T. and W.~U. Bajwa (2019).
\newblock Exsis: Extended sure independence screening for ultrahigh-dimensional
  linear models.
\newblock {\em Signal Processing\/}~{\em 159}, 33--48.

\bibitem[\protect\citeauthoryear{Ando and Zhang}{Ando and Zhang}{2005}]{Ando05}
Ando, R.~K. and T.~Zhang (2005).
\newblock A framework for learning predictive structures from multiple tasks
  and unlabeled data.
\newblock {\em Journal of Machine Learning Research\/}~{\em 6}, 1817--1853.

\bibitem[\protect\citeauthoryear{Banerjee, Mukherjee, and Sun}{Banerjee
  et~al.}{2018}]{Ban19}
Banerjee, T., G.~Mukherjee, and W.~Sun (2018).
\newblock Adaptive sparse estimation with side information.
\newblock {\em arXiv:1811.11930\/}.

\bibitem[\protect\citeauthoryear{Bastani}{Bastani}{2018}]{Bastani18}
Bastani, H. (2018).
\newblock Predicting with proxies: Transfer learning in high dimension.

\bibitem[\protect\citeauthoryear{B{\"u}hlmann and van~de Geer}{B{\"u}hlmann and
  van~de Geer}{2015}]{Bu15}
B{\"u}hlmann, P. and S.~van~de Geer (2015).
\newblock High-dimensional inference in misspecified linear models.
\newblock {\em Electronic Journal of Statistics\/}~{\em 9\/}(1), 1449--1473.

\bibitem[\protect\citeauthoryear{Cai, Sun, and Wang}{Cai et~al.}{2019}]{Tony19}
Cai, T.~T., W.~Sun, and W.~Wang (2019).
\newblock {CARS: C}ovariate-assisted ranking and screening for large-scale
  two-sample inference.
\newblock {\em Journal of the Royal Statistical Society: Series B (Statistical
  Methodology)\/}~{\em 81\/}(2), 187--234.

\bibitem[\protect\citeauthoryear{Cai and Wei}{Cai and Wei}{2019}]{CW19}
Cai, T.~T. and H.~Wei (2019).
\newblock Transfer learning for nonparametric classification: Minimax rate and
  adaptive classifier.
\newblock {\em arXiv:1906.02903\/}.

\bibitem[\protect\citeauthoryear{Candes and Tao}{Candes and Tao}{2007}]{CT07}
Candes, E. and T.~Tao (2007).
\newblock The dantzig selector: Statistical estimation when p is much larger
  than n.
\newblock {\em The annals of Statistics\/}~{\em 35\/}(6), 2313--2351.

\bibitem[\protect\citeauthoryear{Dai, Han, Yang, and Zhang}{Dai
  et~al.}{2018}]{Dai18}
Dai, D., L.~Han, T.~Yang, and T.~Zhang (2018).
\newblock {Bayesian Model Averaging with Exponentiated Least Squares Loss}.
\newblock {\em IEEE Transactions on Information Theory\/}~{\em 64\/}(5),
  3331--3345.

\bibitem[\protect\citeauthoryear{Dai, Rigollet, and Zhang}{Dai
  et~al.}{2012}]{DRZ12}
Dai, D., P.~Rigollet, and T.~Zhang (2012).
\newblock Deviation optimal learning using greedy $ q $-aggregation.
\newblock {\em The Annals of Statistics\/}~{\em 40\/}(3), 1878--1905.

\bibitem[\protect\citeauthoryear{Daum{\'e}~III}{Daum{\'e}~III}{2007}]{Daume07}
Daum{\'e}~III, H. (2007).
\newblock Frustratingly easy domain adaptation.
\newblock In {\em Proceedings of the 45th Annual Meeting of the Association of
  Computational Linguistics}, pp.\  256--263.

\bibitem[\protect\citeauthoryear{Fan and Li}{Fan and Li}{2001}]{FL01}
Fan, J. and R.~Li (2001).
\newblock Variable selection via nonconcave penalized likelihood and its oracle
  properties.
\newblock {\em Journal of the American statistical Association\/}~{\em
  96\/}(456), 1348--1360.

\bibitem[\protect\citeauthoryear{Fan and Lv}{Fan and Lv}{2008}]{FL08}
Fan, J. and J.~Lv (2008).
\newblock Sure independence screening for ultrahigh dimensional feature space.
\newblock {\em Journal of the Royal Statistical Society: Series B (Statistical
  Methodology)\/}~{\em 70\/}(5), 849--911.

\bibitem[\protect\citeauthoryear{Hu, Li, Lu, et~al.}{Hu et~al.}{2019}]{Hu19}
Hu, Y., M.~Li, Q.~Lu, et~al. (2019).
\newblock A statistical framework for cross-tissue transcriptome-wide
  association analysis.
\newblock {\em Nature genetics\/}~{\em 51\/}(3), 568--576.

\bibitem[\protect\citeauthoryear{Johnson-L{\'e}ger, Aurrand-Lions,
  Beltraminelli, et~al.}{Johnson-L{\'e}ger et~al.}{2002}]{JAM2}
Johnson-L{\'e}ger, C.~A., M.~Aurrand-Lions, N.~Beltraminelli, et~al. (2002).
\newblock Junctional adhesion molecule-2 (jam-2) promotes lymphocyte
  transendothelial migration.
\newblock {\em Blood, The Journal of the American Society of Hematology\/}~{\em
  100\/}(7), 2479--2486.

\bibitem[\protect\citeauthoryear{Liu and Kozubowski}{Liu and
  Kozubowski}{2015}]{LK15}
Liu, Y. and T.~J. Kozubowski (2015).
\newblock A folded laplace distribution.
\newblock {\em Journal of Statistical Distributions and Applications\/}~{\em
  2\/}(1), 1--17.

\bibitem[\protect\citeauthoryear{Lounici, Pontil, and Tsybakov}{Lounici
  et~al.}{2009}]{LMT09}
Lounici, K., M.~Pontil, and A.~B. Tsybakov (2009).
\newblock Taking advantage of sparsity in multi-task learning.
\newblock {\em arXiv:0903.1468\/}.

\bibitem[\protect\citeauthoryear{Mao, Chen, and Wong}{Mao et~al.}{2019}]{Mao19}
Mao, X., S.~X. Chen, and R.~K. Wong (2019).
\newblock Matrix completion with covariate information.
\newblock {\em Journal of the American Statistical Association\/}~{\em
  114\/}(525), 198--210.

\bibitem[\protect\citeauthoryear{Mei, Fei, and Zhou}{Mei et~al.}{2011}]{Mei11}
Mei, S., W.~Fei, and S.~Zhou (2011).
\newblock Gene ontology based transfer learning for protein subcellular
  localization.
\newblock {\em BMC bioinformatics\/}~{\em 12\/}(1), 44.

\bibitem[\protect\citeauthoryear{Pan and Yang}{Pan and Yang}{2013}]{Pan13}
Pan, W. and Q.~Yang (2013).
\newblock Transfer learning in heterogeneous collaborative filtering domains.
\newblock {\em Artificial intelligence\/}~{\em 197}, 39--55.

\bibitem[\protect\citeauthoryear{Raskutti, Wainwright, and Yu}{Raskutti
  et~al.}{2011}]{Raskutti11}
Raskutti, G., M.~J. Wainwright, and B.~Yu (2011).
\newblock Minimax rates of estimation for high-dimensional linear regression
  over {$\ell_q $}-balls.
\newblock {\em IEEE transactions on information theory\/}~{\em 57\/}(10),
  6976--6994.

\bibitem[\protect\citeauthoryear{Rigollet and Tsybakov}{Rigollet and
  Tsybakov}{2011}]{RT11}
Rigollet, P. and A.~Tsybakov (2011).
\newblock Exponential screening and optimal rates of sparse estimation.
\newblock {\em The Annals of Statistics\/}~{\em 39\/}(2), 731--771.

\bibitem[\protect\citeauthoryear{Shin, Roth, Gao, et~al.}{Shin
  et~al.}{2016}]{Shin16}
Shin, H.-C., H.~R. Roth, M.~Gao, et~al. (2016).
\newblock Deep convolutional neural networks for computer-aided detection: Cnn
  architectures, dataset characteristics and transfer learning.
\newblock {\em IEEE transactions on medical imaging\/}~{\em 35\/}(5),
  1285--1298.

\bibitem[\protect\citeauthoryear{Sun and Zhang}{Sun and
  Zhang}{2012}]{scaled-lasso}
Sun, T. and C.-H. Zhang (2012).
\newblock Scaled sparse linear regression.
\newblock {\em Biometrika\/}~{\em 99\/}(4), 879--898.

\bibitem[\protect\citeauthoryear{Sun and Hu}{Sun and Hu}{2016}]{Sun16}
Sun, Y.~V. and Y.-J. Hu (2016).
\newblock Integrative analysis of multi-omics data for discovery and functional
  studies of complex human diseases.
\newblock In {\em Advances in genetics}, Volume~93, pp.\  147--190.

\bibitem[\protect\citeauthoryear{Tibshirani}{Tibshirani}{1996}]{Lasso}
Tibshirani, R. (1996).
\newblock Regression shrinkage and selection via the lasso.
\newblock {\em Journal of the Royal Statistical Society: Series B
  (Methodological)\/}~{\em 58\/}(1), 267--288.

\bibitem[\protect\citeauthoryear{Torrey and Shavlik}{Torrey and
  Shavlik}{2010}]{Torrey10}
Torrey, L. and J.~Shavlik (2010).
\newblock Transfer learning.
\newblock In {\em Handbook of research on machine learning applications and
  trends: algorithms, methods, and techniques}, pp.\  242--264. IGI Global.

\bibitem[\protect\citeauthoryear{Tsybakov}{Tsybakov}{2014}]{TS14}
Tsybakov, A.~B. (2014).
\newblock Aggregation and minimax optimality in high-dimensional estimation.
\newblock In {\em Proceedings of the International Congress of Mathematicians},
  Volume~3, pp.\  225--246.

\bibitem[\protect\citeauthoryear{Turki, Wei, and Wang}{Turki
  et~al.}{2017}]{Turki17}
Turki, T., Z.~Wei, and J.~T. Wang (2017).
\newblock Transfer learning approaches to improve drug sensitivity prediction
  in multiple myeloma patients.
\newblock {\em IEEE Access\/}~{\em 5}, 7381--7393.

\bibitem[\protect\citeauthoryear{Verzelen}{Verzelen}{2012}]{Ver12}
Verzelen, N. (2012).
\newblock {Minimax risks for sparse regressions: Ultra-high dimensional
  phenomenons}.
\newblock {\em Electronic Journal of Statistics\/}~{\em 6}, 38--90.

\bibitem[\protect\citeauthoryear{Wang, Shi, Wu, and Ma}{Wang
  et~al.}{2019}]{Wang19}
Wang, S., X.~Shi, M.~Wu, and S.~Ma (2019).
\newblock {Horizontal and vertical integrative analysis methods for mental
  disorders omics data}.
\newblock {\em Scientific Reports\/}, 1--12.

\bibitem[\protect\citeauthoryear{Weiss, Khoshgoftaar, and Wang}{Weiss
  et~al.}{2016}]{weiss2016survey}
Weiss, K., T.~M. Khoshgoftaar, and D.~Wang (2016).
\newblock A survey of transfer learning.
\newblock {\em Journal of Big Data\/}~{\em 3\/}(1), 9.

\bibitem[\protect\citeauthoryear{Xia, Cai, and Sun}{Xia
  et~al.}{2020}]{Xia2020GAP}
Xia, Y., T.~T. Cai, and W.~Sun (2020).
\newblock {GAP}: {A} general framework for information pooling in two-sample
  sparse inference.
\newblock {\em Journal of the American Statistical Association\/}~(to appear).

\bibitem[\protect\citeauthoryear{Zhang}{Zhang}{2010}]{Zhang10}
Zhang, C.-H. (2010).
\newblock Nearly unbiased variable selection under minimax concave penalty.
\newblock {\em The Annals of statistics\/}~{\em 38\/}(2), 894--942.

\bibitem[\protect\citeauthoryear{Zhou}{Zhou}{2009}]{Zhou09}
Zhou, S. (2009).
\newblock Restricted eigenvalue conditions on subgaussian random matrices.
\newblock {\em arXiv:0912.4045\/}.

\bibitem[\protect\citeauthoryear{Zou}{Zou}{2006}]{Zou06}
Zou, H. (2006).
\newblock The adaptive lasso and its oracle properties.
\newblock {\em Journal of the American statistical association\/}~{\em
  101\/}(476), 1418--1429.

\end{thebibliography}
\bibliographystyle{chicago}

\appendix

\section{Proofs in Section \ref{sec2}}
Let $S$ denote the support of $\beta$.
Let $\widehat{\Sig}^{\mA}=\sum_{k \in \mA\cup\{0\}}\alpha_k\widehat{\Sig}^{(k)}$ where $\widehat{\Sig}^{(k)}=(X^{(k)})^{\intercal}X^{(k)}/n_k$. 
As Theorem \ref{thm0-l1} is a special case of Theorem \ref{thm1-l1}, we only need to prove Theorem \ref{thm1-l1}.

Let $w^{\mA}=\beta+\delta^{\mA}$ for 
\[
   \delta^{\mA}=(\sum_{k \in \mA\cup\{0\}}\alpha_k \Sig^{(k)})^{-1}\sum_{k \in \mA\cup\{0\}}\alpha_k \Sig^{(k)}\delta^{(k)}.
\]
\begin{lemma}
\label{tlem1}
Under Condition \ref{cond1}, for some positive $r_1$ and $r_2$ such that $r_1(\log p/(n_{\mA}+n_0))^{1/4}=o(1)$ and $r_2(\log p/n_0)^{1/4}=o(1)$,  with probability at least $1-c_1\exp(-c_2n_0)$,
\[
  \min\left\{ \inf_{0\neq u\in \mathcal{B}_1(r_1)} \frac{u^{\intercal}\widehat{\Sig}^{\mA}u}{\|u\|_2^2},\inf_{0\neq u\in \mathcal{B}_1(r_2)} \frac{u^{\intercal}\widehat{\Sig}^{(0)}u}{\|u\|_2^2}\right\}\geq \phi_0
\]
  for some positive constant $\phi_0>0$.
\end{lemma}
\begin{proof}[Proof of Lemma \ref{tlem1}]
The proof follows from Lemma 1 in \citet{Raskutti11}.
\end{proof}

\begin{lemma}
\label{tlem2}
Under the conditions of Theorem \ref{thm1-l1}, we have for $\hat{u}^{\mA}=\hat{w}^{\mA}-w^{\mA}$,
\begin{align*}
 ( \hat{u}^{\mA})^{\intercal}\widehat{\Sig}^{\mA}\hat{u}^{\mA}\vee \|\hat{u}^{\mA}\|_2^2=O_P(s\lam_w^2+\lam_wC_{\Sig}h)~\text{and}~~\|\hat{u}^{\mA}\|_1=O_P(s\lam_w+ C_{\Sig}h).
\end{align*}
\end{lemma}
\begin{proof}[Proof of Lemma \ref{tlem2}]
In the event that
\begin{align*}
   E_2&=\left\{\frac{1}{n_{\mA}+n_0}\|\sum_{k \in \mA\cup\{0\}}(X^{(k)})^{\intercal}(y^{(k)}-X^{(k)}w^{\mA})\|_{\infty}\leq \frac{\lam_w}{2}, \right.\\
   &\quad\quad \left.\inf_{0\neq 6\|u_S\|_1\geq \|u_{S^c}\|_1}\frac{u^{\intercal}\widehat{\Sig}^{\mA}u}{\|u_S\|_2^2}\geq \phi_0,\inf_{u\in\mathcal{B}_1(C_{\Sig}h)} \frac{u^{\intercal}\widehat{\Sig}^{\mA}u}{\|u\|_2^2}\geq \phi_0\right\},
\end{align*}
we have for $\hat{u}^{\mA}=\hat{w}^{\mA}-w^{\mA}$,
\begin{align*}
\frac{1}{2}(\hat{u}^{\mA})^{\intercal}\widehat{\Sig}^{\mA}\hat{u}^{\mA}&\leq \lam_w\|w^{\mA}\|_1-\lam_w\|\hat{w}^{\mA}\|_1\\
&+|(\hat{u}^{\mA})^{\intercal}\sum_{k \in \mA}(X^{(k)})^{\intercal}(y^{(k)}-X^{(k)}\hat{\delta}^{(k)}-X^{(k)}w^{\mA})|\\
&\leq  \lam_w\|w^{\mA}\|_1-\lam_w\|\hat{w}^{\mA}\|_1+\frac{\lam_w}{2}\|\hat{u}^{\mA}\|_1,
\end{align*}
where the last step is due to event $E_2$.
As a result,
\begin{align*}
\frac{1}{2}(\hat{u}^{\mA})^{\intercal}\widehat{\Sig}^{\mA}\hat{u}^{\mA}&\leq \frac{3}{2}\lam_w\|\hat{u}^{\mA}_S\|_1+\lam_w\|w^{\mA}_{S^c}\|_1-\lam_w\|\hat{w}^{\mA}_{S^c}\|_1+\frac{\lam_w}{2}\|\hat{u}^{\mA}_{S^c}\|_1.
\end{align*}

Using the fact that
\[
  \|\hat{w}^{\mA}_{S^c}\|_1\geq \|\hat{u}^{\mA}_{S^c}\|_1-\|w^{\mA}_{S^c}\|_1,
\]
we arrive at
\[
  \frac{1}{2}(\hat{u}^{\mA})^{\intercal}\widehat{\Sig}^{\mA}\hat{u}^{\mA}\leq \frac{3}{2}\lam_w\|\hat{u}^{\mA}_S\|_1+2\lam_w\|w^{\mA}_{S^c}\|_1-\frac{\lam_w}{2}\|\hat{u}^{\mA}_{S^c}\|_1.
\]
(i) If $ \frac{3}{2}\lam_w\|\hat{u}^{\mA}_S\|_1\geq 2\lam_w\|w^{\mA}_{S^c}\|_1$,
\[
  \frac{1}{2}(\hat{u}^{\mA})^{\intercal}\widehat{\Sig}^{\mA}\hat{u}^{\mA}\leq 3\lam_w\|\hat{u}^{\mA}_S\|_1-\frac{\lam_w}{2}\|\hat{u}^{\mA}_{S^c}\|_1.
\]
Under the restricted eigenvalue condition in $E_2$, for some sufficiently large constants $C_1$ and $C_2$.
\begin{align}
\label{eq6-pf0}
   \|\hat{u}^{\mA}\|_2^2\leq C_1\frac{s\lam_w^2}{\phi_0}~\text{and}~~\|\hat{u}^{\mA}\|_1\leq C_2s\lam_w.
\end{align}
(ii) If $ \frac{3}{2}\lam_w\|\hat{u}^{\mA}_S\|_1\leq 2\lam_w\|w^{\mA}_{S^c}\|_1$,
\[
  \frac{1}{2}(\hat{u}^{\mA})^{\intercal}\widehat{\Sig}^{\mA}\hat{u}^{\mA}\leq 2\lam_w\|w^{\mA}_{S^c}\|_1-\frac{\lam_w}{2}\|\hat{u}^{\mA}_{S^c}\|_1.
\]
Therefore,
\[
    \|\hat{u}^{\mA}\|_1\leq 4\|w^{\mA}_{S^c}\|_1+\frac{4}{3}\|w^{\mA}_{S^c}\|_1\leq \frac{16}{3}\|\delta^{\mA}_{S^c}\|_1\leq \frac{16}{3}C_{\Sig}h.
\]
Under the restricted eigenvalue condition in $E_2$, for some sufficiently large constants $C_1$ and $C_2$, together with (\ref{eq6-pf0}),
\begin{align}
\label{eq6-pf1}
   \|\hat{u}^{\mA}\|_2^2\leq \frac{C_1}{\phi_0}(s\lam_w^2+\lam_wC_{\Sig}h)~\text{and}~~\|\hat{u}^{\mA}\|_1\leq C_2(s\lam_w+ C_{\Sig}h).
\end{align}

It is left to verify that $\P(E_2)\rightarrow 1$.  
Notice that
\begin{align*}
&\frac{1}{n_{\mA}+n_0}\|\sum_{k \in \mA\cup\{0\}}(X^{(k)})^{\intercal}(y^{(k)}-X^{(k)}w^{\mA})\|_{\infty}\\
&=\|\frac{1}{n_{\mA}+n_0}\sum_{k \in \mA\cup\{0\}}(X^{(k)})^{\intercal}\eps^{(k)}+\sum_{k\in \mA\cup\{0\}}\alpha_k\widehat{\Sig}^{(k)}w^{(k)}-\widehat{\Sig}^{\mA_0}w^{\mA})\|_{\infty}\\
&\leq \|\frac{1}{n_{\mA}+n_0}\sum_{k \in \mA\cup\{0\}}(X^{(k)})^{\intercal}\eps^{(k)}+\sum_{k\in \mA\cup\{0\}}\alpha_k(\widehat{\Sig}^{(k)}-\Sig^{(k)})\delta^{(k)}\|_{\infty} +\|(\widehat{\Sig}^{\mA}-\Sig^{\mA})w^{\mA}\|_{\infty}.
\end{align*}
By Condition \ref{cond2}, we have
\[
  \P\left(\frac{1}{n_{\mA}+n_0}\|\sum_{k \in \mA\cup\{0\}}(X^{(k)})^{\intercal}\eps^{(k)}\|_{\infty}>c_1\max_{k\in \mA\cup \{0\}}\sig_k\sqrt{\log p/(n_{\mA}+n_0)}\right)\leq 2/p.
\]
Since $x_{i,j}^{(k)}(x_{i}^{(k)})^{\intercal}\delta^{(k)}$ is sub-exponential, we have 
\begin{align*}
&\P\left( \|\sum_{k\in \mA}\alpha_k(\widehat{\Sig}^{(k)}-\Sig^{(k)})\delta^{(k)}\|_{\infty}\geq  t\right)\\
&\leq 2p\max_{j\leq p,k\in\mA}\exp\left\{-c_1\min\left(\frac{(n_{\mA}+n_0)t^2}{4\|(\Sig^{(k)})^{1/2}\delta^{(k)}\|_2^2},\frac{(n_{\mA}+n_0)t}{2\|(\Sig^{(k)})^{1/2}\delta^{(k)}\|_2}\right)\right\}
\end{align*}
for some constant $c_1>0$. For $t \geq C\max_{j\leq p,k\in\mA}\|(\Sig^{(k)})^{1/2}\delta^{(k)}\|_2\sqrt{\log p/(n_{\mA}+n_0)}$ with a large enough constant $C$ and $\log p=o(n_{\mA}+n_0)$, 
\[
   \P\left( \|\sum_{k\in \mA\cup\{0\}}\alpha_k(\widehat{\Sig}^{(k)}-\Sig^{(k)})\delta^{(k)}\|_{\infty}\geq  t\right)\leq 2/p.
\]
We can similarly show that
\[
  \P\left(\|(\widehat{\Sig}^{\mA}-\Sig^{\mA})w^{\mA}\|_{\infty}\geq c\max_{k\in \mA\cup\{0\}}\sqrt{(w^{(k)})^{\intercal}\Sig^{(k)}w^{(k)}\log p/(n_{\mA}+n_0)}\right)\leq 2/p
\]
for a large enough constant $c$. Since $\max_{k\in \mA\cup\{0\}}(w^{(k)})^{\intercal}\Sig^{(k)}w^{(k)}\leq \E[(y^{(k)}_i)^2]$ and 
\[ \max_{k\in \mA\cup\{0\}}(\delta^{(k)})^{\intercal}\Sig^{(k)}\delta^{(k)}\leq 2\E[(y^{(k)}_i)^2]+2\E[(y^{(0)}_i)^2],
\] it suffices to take $\lam_w\geq c\max_{k\in \mA\cup\{0\}}\sqrt{\E[(y^{(k)}_i)^2]\log p/(n_{\mA}+n_0)}$
for a sufficiently large constant $c$. 

The last two statements hold by Lemma \ref{tlem1} and \citet{Zhou09} under the sample size condition of Lemma \ref{thm1-l1}.
\end{proof}

\begin{proof}[Proof of Theorem \ref{thm0-l1} and Theorem \ref{thm1-l1}]
For some large enough constant $c$, the exists constant $\phi_0\geq 0$ such that
\begin{align*}
E_2'=&\left\{\frac{1}{n_0}\|(X^{(0)})^{\intercal}\eps^{(0)}\|_{\infty}\leq \frac{\lam_{\delta}}{2}, ~\inf_{0\neq u\in\mathcal{B}_1(6C_{\Sig}h)}\frac{u^{\intercal}\widehat{\Sig}^{(0)}u}{\|u\|_2^2}\geq \phi_0,\right.\\
&\quad\quad\left.\inf_{0\neq u\in\mathcal{B}\left(\frac{c}{\lam_{\delta}}(\frac{s\log p}{n_{\mA}+n_0}+C_{\Sig}h\sqrt{\frac{\log p}{n_{\mA}+n_0}})\right)}\frac{u^{\intercal}\widehat{\Sig}^{(0)}u}{\|u\|_2^2}\geq \phi_0\right\}\cap E_2.
\end{align*}

The following oracle inequality holds for $\hat{v}^{\mA}=\hat{\delta}^{\mA}-\delta^{\mA}$.
\begin{align*}
\frac{1}{2n_0}\|X^{(0)}\hat{v}^{\mA}\|_2^2&\leq \lam_{\delta}\|\delta^{\mA}\|_1-\lam_{\delta}\|\hat{\delta}^{\mA}\|_1 +\frac{1}{n_0}|\langle X^{(0)}\hat{v}^{\mA},\eps^{(0)}-X^{(0)}\hat{u}^{\mA}\rangle|.
\end{align*}
For the last term, it holds that , in $E_2'$,
\begin{align*}
\frac{1}{n_0}|\langle X^{(0)}\hat{v}^{\mA},\eps^{(0)}-X^{(0)}\hat{u}^{\mA}\rangle|&\leq
\frac{\lam_{\delta}}{2}\|\hat{v}^{\mA}\|_1+\frac{1}{n_0}\|X^{(0)}\hat{u}^{\mA}\|_2^2+\frac{1}{4n_0}\|X^{(0)}\hat{v}^{\mA}\|_2^2,
\end{align*}
where we use the fact that $|ab|\leq \frac{ca^2}{2}+\frac{b^2}{2c}$ for every $c>0$.
We arrive the following oracle inequality:
\begin{align*}
\frac{1}{4n_0}\|X^{(0)}\hat{v}^{\mA}\|_2^2&\leq  \frac{3}{2}\lam_{\delta}\|\delta^{\mA}\|_1-\frac{1}{2}\lam_{\delta}\|\hat{v}^{\mA}\|_1+\frac{1}{n_0}\|X^{(0)}\hat{u}^{\mA}\|_2^2.
\end{align*}

(i) If  $\frac{3}{2}\lam_{\delta}\|\delta^{\mA}\|_1\geq \frac{1}{n_0}\|X^{(0)}\hat{u}^{\mA}\|_2^2$,
then $\|\hat{v}^{\mA}\|_1\leq 6\|\delta^{\mA}\|_1\leq 6C_{\Sig}h$ and
\begin{align*}
\frac{1}{4n_0}\|X^{(0)}\hat{v}^{\mA}\|_2^2&\leq 3\lam_{\delta}\|\delta^{\mA}\|_1.
\end{align*}
Under the restricted eigenvalue condition in event $E_2'$, standard arguments lead to 
\begin{align*}
 &  \frac{1}{n_0}\|X^{(0)}\hat{v}^{\mA}\|_2^2\leq 12\lam_{\delta}C_{\Sig}h~~\text{and}~\|\hat{v}^{\mA}\|_2^2\leq \frac{c_1C_{\Sig}h\lam_{\delta}}{\phi_0}
\end{align*}
for some constant $c_1>0$.
 
 (ii) If  $\frac{3}{2}\lam_{\delta}\|\delta^{\mA}\|_1\leq \frac{1}{n_0}\|X^{(0)}\hat{u}^{\mA}\|_2^2$, then
\[
\lam_{\delta}\|\hat{v}^{\mA}\|_1\leq \frac{4}{n_0}\|X^{(0)}\hat{u}^{\mA}\|_2^2~\text{and}~ \frac{1}{n_0}\|X^{(0)}\hat{v}^{\mA}\|_2^2\leq  \frac{8}{n_0}\|X^{(0)}\hat{u}^{\mA}\|_2^2.
\]
Using Lemma \ref{tlem2}, we know that in $E_2'$,
\[
   \| \hat{v}^{\mA}\|_1\leq \frac{c}{\lam_{\delta}}(\frac{s\log p}{n_{\mA}+n_0}+C_{\Sig}h\sqrt{\frac{\log p}{n_{\mA}+n_0}}).
\]
Using the restricted eigenvalue conditions in $E_2'$, we arrive at desired results.

Next, we show that $\P(E_2')\rightarrow 1$. Under the sample size condition in Theorem \ref{thm1-l1} for $\lam_{\delta}=c_1\sqrt{\log p/n_0}$,
\[
 C_{\Sig}h(\frac{\log p}{n_0})^{1/4}=o(1)~\text{and}~\frac{c}{\lam_{\delta}}(\frac{s\log p}{n_{\mA}+n_0}+C_{\Sig}h\sqrt{\frac{\log p}{n_{\mA}+n_0}})(\frac{\log p}{n_0})^{1/4}=o(1).
\]
By Lemma \ref{tlem1}, the restricted eigenvalue conditions  in $E_2'$ hold with high probability.
Together with the sub-Gaussian property of $\eps^{(k)}$, we can find
\[
  \P\left(\frac{1}{n_0}\|(X^{(0)})^{\intercal}\eps^{(0)}\|_{\infty}\leq \sqrt{2\log p/n_0}\right)\geq 1-2/p.
\]

\end{proof}

\subsection{Minimax lower bounds for estimation and prediction}

\begin{proof}[Proof of Theorem \ref{thm-mini2}]

When $\frac{s\log p}{n_{\mA_0}+n_0}\geq Ch\frac{\log p}{n_0}$, we consider a special case where $h=0$. That is $y^{(k)}_i=(x^{(k)}_i)^{\intercal}\beta+\eps^{(k)}_i$ for all $k \in \mA_0\cup\{0\}$. We also consider the case where $\sig^2_k=\sig^2_0$ for $k \in \mA_0\cup\{0\}$. In this case, \citet{Raskutti11} consider a class of $\beta$ such that $\|\beta\|_0\leq s$ and $|\beta_j|\in\{0,C\sqrt{\log p/(n_{\mA_0}+n_0)}\}$ and show that
\begin{align*}
\P\left(\min_{\hat{\beta}}\max_{\Theta_q(s,0)} \|\hat{\beta}-\beta\|_2^2\geq c\frac{s\log p}{n_{\mA_0}+n_0}\right)\geq \frac{1}{2}.
\end{align*}

When $\frac{s\log p}{n_{\mA_0}+n_0}\leq Ch\frac{\log p}{n_0}$, we consider $\beta$ in $\mathcal{B}_0(h)$ and $w^{(k)}=0$ for all $1\leq k\leq K$. That is, all the auxiliary samples contain no information about $\beta$. Above results can be applied again. 
\end{proof}
\begin{proof}[Proof of Theorem \ref{thm2-low} and Theorem \ref{thm-mini3}]
Consider a fixed $q\in (0,1]$.
First consider $h=0$ and the lower bound $s\log p/(n_{\mA_q}+n_0)$ follows from the first part of the proof of Theorem \ref{thm-mini2}.

Next, we consider $\beta$ in $\mathcal{B}_0(s)\cap \mathcal{B}_q(h)$ and $w^{(k)}=0$ for all $1\leq k\leq K$. That is, all the auxiliary samples contain no information about $\beta$. The proof for this case follows from the construction in Theorem 5.1 of  \citet{RT11}. We layout an outline of the proof as follows.
Let $\bar{m}$ be the integer part of $h^q(\log p/n_0)^{-q/2}$. 

(i)
If $\bar{m}\geq 1$  and $\frac{s\log p}{n_0}\leq h^q\left(\frac{\log p}{n_0}\right)^{1-q/2}$, then $h\geq (\log p/n_0)^{1/2}$ and $h^q\geq s(\log p/n_0)^{q/2}$. Consider the class of $\beta$ such that $\|\beta\|_0\leq s$ and $|\beta_j|\in\{0,C_2\sqrt{\log p/n_0}\}$ for some $s>0$. Notice that $\|\beta\|_q\leq C_2s^{1/q}(\log p/n_0)^{1/2}\leq C_2h$ for such $\beta$. We can again apply the proof in \citet{Raskutti11} to show that
\begin{align*}
\P\left(\min_{\hat{\beta}}\max_{\beta\in  \mathcal{B}_0(s)\cap \mathcal{B}_q(C_2s^{1/q}(\log p/n_0)) } \|\hat{\beta}-\beta\|_2^2\geq C_3\frac{s\log p}{n_0}\right)\geq \frac{1}{2}.
\end{align*}

(ii) If $\bar{m}\geq 1$ and $\frac{s\log p}{n_0}> h^q\left(\frac{\log p}{n_0}\right)^{1-q/2}$, i.e. $s>\bar{m}$ and $h\geq (\log p/n_0)^{1/2}$,
then consider the class of $\beta$ such that $\|\beta\|_0\leq \bar{m}/2$ and $|\beta_j|\in\{0,C_2\sqrt{\log p/n_0}\}$. Notice that $\|\beta\|_q\leq C_2\bar{m}^{1/q}(\log p/n_0)^{1/2}) \leq C_2h$. Above prove can be again used to show that 
\begin{align*}
\P\left(\min_{\hat{\beta}}\max_{\beta\in  \mathcal{B}_0(\bar{m}/2)\cap \mathcal{B}_q(C_2\bar{m}^{1/q}(\log p/n_0)^{1/2}) } \|\hat{\beta}-\beta\|_2^2\geq C_3\frac{\bar{m}\log p}{n_0}\right)\geq \frac{1}{2},
\end{align*}
where $\bar{m}\log p/n_0 \geq h^q(\log p/n_0)^{1-q/2}/2$.

(iii) If $\bar{m}=0$ and $s\geq 1$, then $h<(\log p/n_0)^{1/2}$ and we consider the class of $\beta$ such that $\|\beta\|_0=1$ and $|\beta_j|\in\{0,h\}$. A similar proof will lead to
\begin{align*}
\P\left(\min_{\hat{\beta}}\max_{\beta\in  \mathcal{B}_0(1)\cap \mathcal{B}_q(h) } \|\hat{\beta}-\beta\|_2^2\geq h^{2}\right)\geq \frac{1}{2}.
\end{align*}
\end{proof}

\section{Proofs of theorems and lemmas in Section \ref{sec3}}
\label{ap-sec3}
\subsection{Proof of Lemma \ref{thm-ag1} and Remark \ref{re1}}
\label{ap-sec3-1}
\begin{proof}[Proof of Lemma \ref{thm-ag1}]
Let $\widehat{B}\in\R^{p\times (L+1)}$ denotes a dictionary such that  $\widehat{B}_{.,l+1}=\hat{\beta}(\widehat{G}_l)$.
The prediction error follows from The proof of Corollary 3.1 in \citet{DRZ12}. We derive the estimation error bound as follows. Let 
\[
    l^*=\argmin_{0\leq l\leq L} \frac{1}{|\mathcal{I}^c|}\|X^{(0)}_{\mathcal{I}^c}(\hat{\beta}(\widehat{G}_l)-\beta)\|_2^2
    \]
    and $\theta^*=e_{l^*}\in\R^{L+1}$.
Using the prediction error bound, we have
\begin{align}
   \frac{1}{|\mathcal{I}^c|}\|X^{(0)}_{\mathcal{I}^c}\widehat{B}(\hat{\theta}-\theta^*)\|_2^2&\leq   \frac{2}{|\mathcal{I}^c|}\|X^{(0)}_{\mathcal{I}^c}(\widehat{B}\hat{\theta}-\beta)\|_2^2 + \frac{2}{|\mathcal{I}^c|}\|X^{(0)}_{\mathcal{I}^c}(\widehat{B}\theta^*-\beta)\|_2^2 \nonumber\\
   &=O_P( \frac{1}{|\mathcal{I}^c|}\|X^{(0)}_{\mathcal{I}^c}(\widehat{B}\theta^*-\beta)\|_2^2+\frac{\log L}{n_0}).\label{eq0-pf-lem1}
\end{align}
We first bound $\|\widehat{B}(\hat{\theta}-\theta^*)\|_2^2$. Let $\widehat{\Sig}^{(0,c)}=\frac{1}{|\mathcal{I}^c|}\sum_{i\in \mathcal{I}^c}x_i^{(0)}(x_i^{(0)})^{\intercal}$.
We have
\begin{align}
\|\widehat{B}(\hat{\theta}-\theta^*)\|_2^2&\leq \frac{\|\Sig^{1/2}\widehat{B}(\hat{\theta}-\theta^*)\|_2^2}{\Lambda_{\min}(\Sig)}\nonumber\\
&\leq \frac{ \frac{2}{|\mathcal{I}^c|}\|X^{(0)}_{\mathcal{I}^c}\widehat{B}(\hat{\theta}-\theta^*)\|_2^2}{\Lambda_{\min}(\Sig)} +  \frac{2}{\Lambda_{\min}(\Sig)}\underbrace{|\langle \widehat{B}(\hat{\theta}-\theta^*),(\widehat{\Sig}^{(0,c)}-\Sig)\widehat{B}(\hat{\theta}-\theta^*)\rangle |}_{R_1},\label{eq1-pf-lem1}
\end{align}
where $\widehat{\Sig}^{(0,c)}=\frac{1}{|\mathcal{I}^c|}(X^{(0)}_{\mathcal{I}^c})^{\intercal}X^{(0)}_{\mathcal{I}^c}$.
For the second term, consider a singular value decomposition of $\widehat{B}$ as
$\widehat{B}=U\Lambda V^{\intercal}$, where $\Lambda$ is a diagonal matrix containing singular values.  We have
\begin{align*}
R_1\leq \|\Lambda V^{\intercal}(\hat{\theta}-\theta^*)\|_2^2\|U^{\intercal}(\widehat{\Sig}^{(0,c)}-\Sig)U\|_2.
\end{align*}
Since $U$ is independent of $X^{(0)}_{\mathcal{I}^c}$ and $X^{(0)}$ is Gaussian distributed, we have
\[
  \P(\|U^{\intercal}(\widehat{\Sig}^{(0,c)}-\Sig)U\|_2\geq c_1\|U^{\intercal}\Sig U\|_2\sqrt{\frac{L}{n_0}})\leq \exp(-c_2 L).
\]
Hence,
\[
  \P(R_1\geq c_1\|\widehat{B}(\hat{\theta}-\theta^*)\|_2^2\|\Sig\|_2\sqrt{\frac{L}{n_0}})\leq \exp(-c_2 L),
  \]
  where we use the fact of singular value decomposition that
  \[
     \|\widehat{B}(\hat{\theta}-\theta^*)\|_2^2=\|\Lambda V^{\intercal}(\hat{\theta}-\theta^*)\|_2^2.
  \]
  When $ c_1\|\Sig\|_2\sqrt{\frac{L}{n_0}}\leq 1/2$, invoking (\ref{eq1-pf-lem1}), we have
  \[
     \|\widehat{B}(\hat{\theta}-\theta^*)\|_2^2=O_P\left(\frac{1}{|\mathcal{I}^c|}\|X^{(0)}_{\mathcal{I}^c}\widehat{B}(\hat{\theta}-\theta^*)\|_2^2\right).
  \]
  Finally,
  \[
   \|\widehat{B}\hat{\theta}-\beta\|_2^2\leq 2 \|\widehat{B}(\hat{\theta}-\theta^*)\|_2^2+ 2\|\widehat{B}\hat{\theta}^*-\beta\|_2^2.
  \]
  Invoking (\ref{eq0-pf-lem1}), we arrive at desired results.
  
\end{proof}

\begin{proof}[Proof of Remark \ref{re1}]
We use (\ref{eq1-pf-lem1}) again. For $R_1$, we consider a different error splitting
\[
  R_1\leq \|\hat{\theta}-\theta^*\|^2_1\|\widehat{B}^{\intercal}(\widehat{\Sig}^{(0,c)}-\Sig)\widehat{B}\|_{\infty}.
\]
Notice that $\|\hat{\theta}-\theta^*\|_1\leq \|\hat{\theta}\|_1+\|\theta^*\|_1\leq 2$. Moreover, $\widehat{B}$ is indepdent of $\widehat{\Sig}^{(0,c)}$. By conditioning on $\widehat{B}$ first, the Gaussian property of $x_i^{(0)}$ gives
\[
  \|\widehat{B}^{\intercal}(\widehat{\Sig}^{(0,c)}-\Sig)\widehat{B}\|_{\infty}=O_P\left(\max_{0\leq l\leq L}(\widehat{B}^{\intercal}\Sig\widehat{B})_{l,l}\sqrt{\frac{\log L}{n_0}}\right)=O_P\left(\sqrt{\frac{\log L}{n_0}}\right)
\]
since $\E[\|y^{(k)}_i\|_2^2]=O(1)$.

\end{proof}

\subsection{Proof of Theorem \ref{sec4-lem1} and Corollary \ref{cor1-l1}}
\label{ap-sec3-2}
\begin{proof}[Proof of Theorem \ref{sec4-lem1} and Corollary \ref{cor1-l1}]
In view of Lemma \ref{thm-ag1} and Theorem \ref{thm0-l1}, it suffices to prove (\ref{cond-agg2}).

\underline{Part (i)}.  We first prove that under Condition \ref{cond4}(a), 
\begin{align}
\label{eq1-sec4-lem1}
& \P\left(  \min_{k \in \mA^c} \|\Sig^{(k)}w^{(k)}-\Sig^{(0)}\beta\|_{2,\hT_k}^2\geq   \min_{k \in \mA^c}\sum_{j\in H_k}(\Sig_{j,.}^{(k)}w^{(k)}-\Sig_{j,.}^{(0)}\beta)^2\right)\nonumber\\
 &\quad\geq 1-O(\exp\{-c_2n_0^{1-2\kappa}/\log(n_0)\}).
\end{align}
By definition, $|H_k|\ll t^*$. This is because, if $|H_k|\gtrsim t_*$, then 
\[
  \|\Sig^{(k)}w^{(k)}-\Sig^{(0)}\beta\|_2^2\geq|H_k| n^{-2\kappa}\gtrsim n^{\alpha-2\kappa} \gg 1,
\]
which is contradictory to $\max_k\E[(y_i^{(k)})^2]\leq C<\infty$. Hence, $|H_k|\leq t_*$ for any $t_*\lesssim n_*$.

Since $X^{(k)}$ and $\eps^{(k)}$ are all Gaussian-distributed, it follows from the proof of Theorem 1 in \citet{FL08} that
\[
   \max_{k\in\mA^c}\P\left(\min_{j\in H_k}|\widehat{\Delta}^{(k)}_{j,.}|<n_*^{-\kappa}~\text{or}~\|\widehat{\Delta}^{(k)}\|_2^2\geq Cp/n_*\right)\leq c_1|H_k|\exp\{-c_2n_*^{1-2\kappa}/\log(n_0)\}
\]
for some constants $C$, $c_0$, and $c_2$.
Taking a uniform over $k\in \mA^c$, we have
\begin{align*}
&\P\left(\min_{j\in H_k}|\widehat{\Delta}^{(k)}_{j,.}|<n_*^{-\kappa}~\text{or}~\|\widehat{\Delta}^{(k)}\|_2^2\geq Cp/n_*,~\forall k\in\mA^c\right)\\
&\quad\leq c_0|H_k||\mA^c|\exp\{-c_2n_*^{1-2\kappa}/\log(n_0)\}.
\end{align*}
Under Condition \ref{cond4}, $\log |\mA^c|\leq c_1\log n_*$.
Therefore, with probability $1-c_0\exp\{-c_2n_*^{1-2\kappa}/\log(n_*)+c_3\log n_*\}$,
\[
   \#\{1\leq j\leq p: |\widehat{\Delta}^{(k)}_j|>cn_*^{-\kappa}\}\leq \frac{p}{n_*^{1-2\kappa}} ~\forall k\in\mA^c.
\]
Define 
\[
  \widehat{T}_k(t)=\{1\leq j\leq p: |\widehat{\Delta}^{(k)}_j|~\text{is among the first }t~\text{largest of all}\}.
  \]
Therefore, for $\gam n_*^{1-2\kappa}\rightarrow \infty$,  we have for large enough $n_*$
\[
   \P(H_k\subseteq \widehat{T}_k(\gam p))=1-O(\exp\{-c_2n_*^{1-2\kappa}/\log(n_*)\}).
\]
Then it follows from the trimming arguments in Step 2 of the proof of Theorem 1 \citep{FL08} that for $t_*\lesssim n_*$,
\begin{align}
\label{eq2-sec4-lem1}
  \P\left(H_k\subseteq \widehat{T}_k,~ \forall~ k\in \mA^c\right)\geq 1-O(\exp\{-c_2n_*^{1-2\kappa}/\log(n_*)\}).
\end{align}
As a result, with probability at least $O(\exp\{-c_2n_*^{1-2\kappa}/\log(n_*)\})$,
\[
 \min_{k \in \mA^c} \|\Sig^{(k)}w^{(k)}-\Sig^{(0)}\beta\|_{2,\hT_k}^2\geq   \min_{k \in \mA^c}\sum_{j\in H_k}(\Sig_{j,.}^{(k)}w^{(k)}-\Sig_{j,.}^{(0)}\beta)^2.
\]
Hence, (\ref{eq1-sec4-lem1}) is proved. 

\underline{Part (ii)}.
For $\widehat{R}^{(k)}$ defined in (\ref{eq-hRk}),
we will show that
\[
   \P\left(\max_{k\in \mA^o}\widehat{R}^{(k)}_1< \min_{k \in\mA^c}\widehat{R}^{(k)}_1\right)\rightarrow 1.
\]
For $k\in\mA^c$, by (\ref{eq2-sec4-lem1})
\begin{align*}
\|\widehat{\Delta}^{(k)}_{\hT_k}\|_2^2\geq \|\widehat{\Delta}^{(k)}_{H_k}\|_2^2.
\end{align*}
Notice that
\begin{align*}
\widehat{\Delta}^{(k)}&=\Sig^{(k)}w^{(k)}-\Sig^{(0)}\beta\\
&\quad + \underbrace{ (X^{(k)})^{\intercal}y^{(k)}/n_k-\E[(X^{(k)})^{\intercal}y^{(k)}/n_k]\}-\{(X^{(0)})^{\intercal}y^{(0)}/n_0-\E[(X^{(0)})^{\intercal}y^{(0)}/n_0]\}}_{E^{(k)}},
\end{align*}
Then for any $0<\delta<1$,
\begin{align*}
\min_{k\in \mA^c}\|\widehat{\Delta}^{(k)}_{H_k}\|_2^2\geq (1-\delta)\min_{k\in \mA^c}\sum_{j\in H_k}\{e_j^T(\Sig^{(k)}w^{(k)}-\Sig^{(0)}\beta)\}^2-\frac{(1-\delta)}{\delta}\max_{k\in \mA^c}\|E^{(k)}_{H_k}\|_2^2.
\end{align*}
As $E^{(k)}_{j}$ is sub-exponential for $1\leq j\leq p$, it is easy to show that
\begin{align*}
    \max_{k\in \mA^c}\|E^{(k)}_{H_k}\|_2^2&\leq \max_{k\in \mA^c}|H_k| \max_{j\in H_k,k\in \mA^c}|E^{(k)}_j|^2,
\end{align*}
where
\begin{align*}
\P\left(\max_{j\in H_k,k\in \mA^c}|E^{(k)}_j|^2\geq x\right)&\leq \max_{k\in \mA^c}|H_k| K\max_{j\in H_k,k\in \mA^c}\P\left(|E^{(k)}_j|>\sqrt{x}\right)\\
&\leq 2t_*K\exp\{-c_1\min\{\frac{n_kx}{c_2},c_3\sqrt{x}\}\}.
\end{align*}
Hence, if $\log (t_*K)\leq c_4\sqrt{n_*}$ for some small enough $c_4$, we have
\[
  \P\left( \max_{k\in \mA^c}\|E^{(k)}_{H_k}\|_2^2\geq ct_*\frac{\log (t_*K)}{n_*}\right)\leq 1/(t_*K).
\]
  Notice that $\log(t_*K)\leq \log n_*+\log K=O(\log p)$. Under Condition \ref{cond4}(b), we have
  \[
    \max_{k\in \mA^c}\|E^{(k)}_{H_k}\|_2^2=o_P(1)\min_{k\in \mA^c}\sum_{j\in H_k}\{e_j^T(\Sig^{(k)}w^{(k)}-\Sig^{(0)}\beta)\}^2
  \]
Therefore, for  any constant $0<\delta<1$,
  \begin{align}
  \label{eq3-sec4-lem1}
  \min_{k\in \mA^c}\|\widehat{\Delta}^{(k)}_{\hT_k}\|_2^2\geq (1-\delta-o_P(1))\min_{k\in \mA^c}\sum_{j\in H_k}\{e_j^T(\Sig^{(k)}w^{(k)}-\Sig^{(0)}\beta)\}^2.
  \end{align}
  
For $k\in\mA$, we have
\[
  \max_{k\in \mA}\|\widehat{\Delta}^{(k)}_{\hT_k}\|_2^2\leq (1+\delta)\max_{k\in \mA}\|\Sig^{(k)}w^{(k)}-\Sig^{(0)}\beta\|_2^2+\frac{(1+\delta)}{\delta}\max_{k\in \mA}\max_{|T|=t_*}\|E^{(k)}_{T}\|_2^2.
\]
For the second part, we can similarly show that for some large enough constant $c$
\begin{align*}
\P\left( \max_{k\in \mA}\max_{|T|=t_*}\|E^{(k)}_{T}\|_2^2\geq c\frac{t_*\log p}{n_*}\right)\leq 1/p
\end{align*}
if $\log p\leq  c_1\sqrt{n_*}$ for some small enough $c_1$.
Hence, by Condition \ref{cond4}(b),
\begin{align}
\label{eq4-sec4-lem1}
\max_{k\in \mA}\|\widehat{\Delta}^{(k)}_{\hT_k}\|_2^2\leq (1+\delta)\max_{k\in \mA}\|\Sig^{(k)}w^{(k)}-\Sig^{(0)}\beta\|_2^2+o_P(1)\min_{k\in \mA^c} \sum_{j\in H_k}\{e_j^T(\Sig^{(k)}w^{(k)}-\Sig^{(0)}\beta)\}^2.
\end{align}
Therefore, by the definition of $\mA^o$,
\begin{align*}
&\P(\max_{k\in \mA^o}\|\widehat{\Delta}^{(k)}_{\hT_k}\|_2^2\geq \min_{k\in \mA^c}\|\widehat{\Delta}^{(k)}_{\hT_k}\|_2^2)\\
&\leq \P\left( (1+\delta)\max_{k\in \mA^o}\|\Sig^{(k)}w^{(k)}-\Sig^{(0)}\beta\|_2^2\geq (1-\delta-o(1))\min_{k\in \mA^c}\sum_{j\in H_k}\{e_j^T(\Sig^{(k)}w^{(k)}-\Sig^{(0)}\beta)\}^2\right)+o(1)\\
&=o(1)
\end{align*}

\end{proof}

\section{Proofs for theorems in Section \ref{sec4}}
\label{ap-ms}
Let $H_k$ denote the support of $\delta^{(k)}$ for $k\in\mA_0$.
\subsection{Proof of Theorem \ref{thm1-l0}}
 The upper bound $s\log p/n_0$ the convergence rate when $\mathcal{A}_0$ is empty. When $\mathcal{A}_0$ is not empty, $s\log p/n_0$ can be trivially achieved since $\mathcal{A}_0$ is nonempty meaning that using some auxiliary samples should be no worse than only using the primary sample. We focus on proving other terms.

\begin{lemma}
\label{lem1-pf1}
Under the conditions of Theorem \ref{thm1-l0}, with probability at least $1-\exp(-c_1\log p)-\exp(-c_2n_{\mA_0})$, it holds that for any $k \in \mA_0$, there exists a sufficiently large constant $c$ such that
\begin{align*}
&\|\hat{\delta}^{(k)}-\delta^{(k)}\|_2^2\leq ch\left(\frac{\log p}{n_0\wedge n_k}\right)\\
&\|\hat{\delta}^{(k)}-\delta^{(k)}\|_1\leq ch\left(\frac{\log p}{n_0\wedge n_k}\right).
\end{align*}
\end{lemma}
\begin{proof}[Proof of Lemma \ref{lem1-pf1}]
 In the event that 
\begin{align*}
   E_0=&\left\{ \|\widehat{\Sig}^{(0)}\beta-\widehat{\Sig}^{(k)}w^{k}-\widehat{\Sig}^{\mA_0}\delta^{(k)}+(X^{(k)})^{\intercal}\eps^{(k)}/n_k-(X^{(0)})^{\intercal}\eps^{(0)}/n_0\|_{\infty}\leq \frac{\lam_k}{2},~\forall k \in \mA_0\right.\\
& \left.\quad \min_{k \in \mA_0}\inf_{0\neq 3\|u_{H_k}\|_1\geq \|u_{H_k^c}\|_1}\frac{u^{\intercal}\widehat{\Sig}^{\mA_0}u}{\|u_{H_k}\|_2^2}\geq \phi_0>0 \right\},
\end{align*}
it is easy to show that
\begin{align}
\label{ora-ineq1}
  \frac{1}{2}(\hat{\delta}^{(k)}-\delta^{(k)})^{\intercal}\widehat{\Sig}^{\mA_0}(\hat{\delta}^{(k)}-\delta^{(k)})\leq \lam_k\|\delta^{(k)}\|_1-\lam_k\|\hat{\delta}^{(k)}\|_1+\frac{\lam_k}{2}\|\hat{\delta}^{(k)}-\delta^{(k)}\|_1.
\end{align}
Since 
\[
   \|\delta^{(k)}_{H_k}\|_1-\|\hat{\delta}^{(k)}_{H_k}\|_1\leq \|(\hat{\delta}^{(k)}-\delta^{(k)})_{H_k}\|_1~\text{and}~ \|\delta^{(k)}_{H^c_k}\|_1-\|\hat{\delta}^{(k)}_{H^c_k}\|_1=- \|(\hat{\delta}^{(k)}-\delta^{(k)})_{H_k^c}\|_1,
\]
we arrive at
\[
  \frac{1}{2}(\hat{\delta}^{(k)}-\delta^{(k)})^{\intercal}\widehat{\Sig}^{\mA_0}(\hat{\delta}^{(k)}-\delta^{(k)})\leq \frac{3\lam_k}{2}\|(\hat{\delta}^{(k)}-\delta^{(k)})_{H_k}\|_1.
\]
Standard arguments lead to  
\begin{align}
\frac{1}{n_k}\|X^{(k)}(\hat{\delta}^{(k)}-\delta^{(k)})\|_2^2\leq \frac{c_1h\lam_k^2}{\phi_0}\nonumber\\
   \|\hat{\delta}^{(k)}-\delta^{(k)}\|_1\leq \frac{c_2h\lam_k}{\phi_0}.\label{re-deltak}
\end{align}

It is left to verify $\P(E_0)\rightarrow 1$. Since $X^{(k)}$ are Gaussian matrices with positive definite covariance matrix, it follows from Theorem 1.6 in \citet{Zhou09} that for $n_{\mA_0}+n_0\gg (s\log p+\log K)$, 
\begin{align*}
    &\P\left(\min_{k \in \mA}\inf_{0\neq 3\|u_{H_k}\|_1\geq \|u_{H_k^c}\|_1}\frac{u^{\intercal}\widehat{\Sig}^{\mA_0}u}{\|u_{H_k}\|_2^2}\geq \Lambda_{\min}(\Sig)/2\right)\\
    &\leq K\exp(-c_1(n_{\mA_0}+n_0))=\exp(-c_2(n_{\mA_0}+n_0)).
\end{align*}
Using the sub-Guassian property of $\eps^{(k)},\eps^{(0)}$ and the sub-exponential property of $\widehat{\Sig}^{(k)}_{j,.}\beta$ and $\widehat{\Sig}^{(0)}_{j,.}\beta$, we can show that for
for \begin{align*}
   \lam_k&\geq  c_1 \sqrt{\frac{\E[(y^{(0)}_1)^2]\log p}{n_0}}+c_1\sqrt{\frac{\E[(y^{(k)}_1)^2]\log p}{n_k}}\leq c_2\sqrt{\frac{\log p}{n_0\wedge n_k}}
\end{align*}
with large enough $c_2>0$, it holds that
\[
  \P(E_0)\geq 1-\exp(-c_1n_{\mA_0})-\exp(-c_2\log p). 
\]
\end{proof}

\begin{proof}[Proof of Theorem \ref{thm1-l0}]
Let $\hat{u}^{\mA_0}=\hat{\beta}(\mA_0)-\beta$.
One can show that
\begin{align}
\label{eq11-pf}
\frac{1}{4(n_{\mA_0}+n_0)}\|X^{(k)}\hat{u}^{\mA_0}\|_2^2\leq \frac{3\lam_{\beta}}{2}\|\hat{u}^{\mA_0}_S\|_1-\frac{\lam_{\beta}}{2}\|\hat{u}^{\mA_0}_{S^c}\|_1+\frac{1}{n_{\mA_0}+n_0}\sum_{k\in \mA_0}\|X^{(k)}(\hat{\delta}^{(k)}-\delta^{(k)})\|_2^2
\end{align}
for $\lam_{\beta}\geq c_1\sqrt{\log p/(n_{\mA_0}+n_0)}$ with large enough constant $c_1$.
By Lemma \ref{lem1-pf1}, we have
\begin{align*}
   \frac{1}{n_{\mA_0}+n_0}\sum_{k\in \mA_0}\|X^{(k)}(\hat{\delta}^{(k)}-\delta^{(k)})\|_2^2&=O_P\left(\sum_{k\in\mA_0:n_k\geq n_0}\frac{n_kh\log p/n_0}{n_{\mA_0}+n_0}+\sum_{k\in\mA_0:n_k< n_0}\frac{h\log p}{n_{\mA_0}+n_0}\right)\\
   &=O_P\left(\frac{h\log p}{n_0}\right),
\end{align*}
where the last step is due to $n_{\mA_0}+n_0\gtrsim |\mA_0|n_0$. By a similar proof of Theorem \ref{thm0-l1}, one can show that
\begin{align*}
\ \frac{1}{n_{\mA_0}+n_0}\sum_{k\in \mA_0}\|X^{(k)}(\hat{\beta}(\mA_0)-\beta)\|_2^2=O_P\left(\frac{s\log p}{n_{\mA_0}+n_0} + \frac{h\log p}{n_0}\right)
\end{align*}
if $s\log p/(n_{\mA_0}+n_0)=o(1)$. For the estimation error, we layout the key steps for the proof. 

If $\frac{3\lam_{\beta}}{2}\|\hat{u}^{\mA_0}_S\|_1\geq \frac{1}{n_{\mA_0}+n_0}\sum_{k\in \mA_0}\|X^{(k)}(\hat{\delta}^{(k)}-\delta^{(k)})\|_2^2$, then under the restricted values conditions, we can show
\[
  \|\hat{\beta}(\mA_0)-\beta)\|_2^2=O_P(\frac{s\log p}{n_{\mA_0}+n_0}).
\]
If $\frac{3\lam_{\beta}}{2}\|\hat{u}^{\mA_0}_S\|_1\leq \frac{1}{n_{\mA_0}+n_0}\sum_{k\in \mA_0}\|X^{(k)}(\hat{\delta}^{(k)}-\delta^{(k)})\|_2^2$, we have $\|\hat{u}^{\mA_0}\|_1\in\mathcal{B}_1(h\log p/n_0/\lam_{\beta})$.
 Under the sample size condition $h\log p/n_0=o((\log p/(n_0+n_{\mA_0}))^{1/4})$, the restricted eigenvalue condition is guaranteed by Lemma \ref{tlem1} and
 \[
   \|\hat{\beta}(\mA_0)-\beta)\|_2^2=O_P(h\log p/n_0).
 \]

\end{proof}

\subsection{Minimax optimal rates for $q\in(0,1)$}
We first prove the minimax lower bound of estimation error in $\Theta_q(s,h)$.
\begin{theoremA}[Minimax lower bound for $q\in(0,1)$]
\label{thm-mini3}
Assume that Condition \ref{cond1} and Condition \ref{cond2} hold true. Suppose that $\max\{h^q(\log p/n_0)^{1/2-q/4}, s\log p/(n_{\mA_0}+n_0)\}=o(1)$.
For any fixed $q\in (0,1)$, there exist some large enough constants $c_1$ and $c_2$,
\begin{align*}
\P\left(\inf_{\hat{\beta}}\sup_{\Theta_q(s,h)} \|\hat{\beta}-\beta\|_2^2\geq c_1\frac{s\log p}{n_{\mA_q}+n_0}+ c_2h^2\wedge h^q\left(\frac{\log p}{n_0}\right)^{1-q/2}\wedge\frac{s\log p}{n_0}\right)\geq \frac{1}{2}.
\end{align*}
\end{theoremA}

For the upper bound, we consider the following algorithm.

\begin{algorithm}[H]
\NoCaptionOfAlgo
 \SetKwInOut{Input}{Input}
    \SetKwInOut{Output}{Output}
\SetAlgoLined
 \Input{Primary data $(X^{(0)},y^{(0)})$ and informative auxiliary samples $\{X^{(k)},y^{(k)}\}_{k\in\mA_q}$}
 \Output{$\hat{\beta}(\mA_q)$}

\underline{Step 1}. 
Estimate each individual $\delta^{(k)}$ via
\begin{align}
\label{eq-hdelta-q}
   \hat{\delta}^{(k)}=\argmin_{\delta\in\R^p}\left\{\frac{1}{2}\delta^{\intercal}\widehat{\Sig}^{\mA_q}\delta-\delta^{\intercal}[(X^{(k)})^{\intercal}y^{(k)}/n_k-(X^{(0)})^{\intercal}y^{(0)}/n_0]+\lam_k\|\delta\|_q\right\},
\end{align}
where $\widehat{\Sig}^{\mA_q}=\sum_{k\in \mA_q\cup\{0\}}(X^{(k)})^{\intercal}X^{(k)}/(n_{\mA_q}+n_0)$ and $\lam_k>0$ are tuning parameters.

\underline{Step 2}. Compute
\begin{align}
\label{eq-hbeta-init}
\hat{\beta}(\mA_q)&=\argmin_{b\in\R^p} \left\{\frac{1}{2(n_{\mA_q}+n_0)}\sum_{k \in \mA_q\cup\{0\}}\|y^{(k)}-X^{(k)}\hat{\delta}^{(k)}-X^{(k)}b\|_2^2+\lam_{\beta}\|b\|_1\right\}\;
\end{align}
for $\lam_{\beta}=c_1\sqrt{\log p/(n_0+n_{\mA_q})}$ with some constant $c_1$.
 \caption{\textbf{The Oracle Trans-Lasso algorithm($\ell_q$)}}
\end{algorithm}

\begin{theoremA}[Achievability of upper bound  for $q\in (0,1)$]
\label{thm2-up}
Assume that Condition \ref{cond1} and Condition \ref{cond2} hold true. Suppose that 
\[
  h(\log p/n_0)^{1/-q/2}=o\left((\frac{\log p}{n_0+n_{\mA_q}})^{1/4}\right),~ \frac{s\log p}{n_{\mA_q}+n_0}=o(1),~ \text{and} ~ n_{\mA_q}\gtrsim |\mA_q|n_0.
  \]
We take $\lam_k\geq c_1(\|y^{(k)}\|_2/n_k+ \|y^{(0)}\|_2/n_0)\sqrt{\log p}$ and $\lam_{\beta}=c_2\sqrt{\log p/(n_0+n_{\mA_0})}$ for sufficiently large constants $c_1$ and $c_2$. Then
\begin{align*}
&\sup_{\beta\in\Theta_q(s,h)} \frac{\sum\limits_{k\in\mA_q\cup\{0\}}\|X^{(k)}(\hat{\beta}(\mA_q)-\beta)\|_2^2 }{n_0+n_{\mA_q}}\vee \|\hat{\beta}(\mA_q)-\beta\|_2^2\\
&=O_P\left( \frac{s\log p}{n_{\mA_q}}+ h^2\wedge h^q\left(\frac{\log p}{n_0}\right)^{1-q/2}\wedge \frac{s\log p}{n_0}\right).
\end{align*}
\end{theoremA}
\begin{proof}[Proof of Theorem \ref{thm2-up}]
Let $\hat{v}^{(k)}=\hat{\delta}^{(k)}-\delta^{(k)}$. Consider the event
\begin{align*}
   E_q=&\left\{ \|\widehat{\Sig}^{(0)}\beta-\widehat{\Sig}^{(k)}w^{k}-\widehat{\Sig}^{\mA_q}\delta^{(k)}+(X^{(k)})^{\intercal}\eps^{(k)}/n_k-(X^{(0)})^{\intercal}\eps^{(0)}/n_0\|_{\infty}\leq \frac{\lam_k}{2},~\forall k \in \mA_q\right.\\
& \left.\quad \min_{k \in \mA}\inf_{0\neq u\in \mathcal{B}_q(3h)}\frac{u^{\intercal}\widehat{\Sig}^{\mA_q}u}{\|u\|_2^2}\geq \phi_0>0 \right\},
\end{align*}
We have
\begin{align*}
\frac{1}{2}(\hat{v}^{(k)})^{\intercal}\widehat{\Sig}^{A_0}\hat{v}^{(k)}&\leq \frac{\lam_k}{2}\|\hat{v}^{(k)}\|_1 + \lam_k\|\delta^{(k)}\|_q-\lam_k\|\hat{\delta}^{(k)}\|_q.
\end{align*}

Since $\|\hat{\delta}^{(k)}\|_q\geq \|\hat{v}^{(k)}\|_q-\|\delta^{(k)}\|_q$, we have $\hat{v}^{(k)}\in B_q(3h)$ and
\[
  \|\hat{v}^{(k)}\|_2^2\leq \phi_0  (\hat{v}^{(k)})^{\intercal}\widehat{\Sig}^{\mA_q}\hat{v}^{(k)}\leq 2\lam_k\|\hat{v}^{(k)}\|_1.
\]
First, $\|\hat{v}^{(k)}\|_2\leq \|\hat{v}^{(k)}\|_q\leq 3h$ for any $q\in(0,1)$.
By Lemma 5 of \citet{Raskutti11} (Notice that $B_q(R_q)$ in \citet{Raskutti11} is defined such that $\|v\|_q^q\leq R_q$), we know that for any $\tau>0$,
\[
   \|\hat{v}^{(k)}\|_1\leq \sqrt{2h^q}\tau^{-q/2}\|\hat{v}^{(k)}\|_2+2h^q\tau^{1-q}.
\]
Let $\tau=2\lam_k/\phi_0$, we obtain
\[
    \|\hat{v}^{(k)}\|_2^2\leq \sqrt{2h^q}\tau^{1-q/2}\|\hat{v}^{(k)}\|_2+2h\tau^{2-q},
\]
which is a quadratic constraint on $\|\hat{v}^{(k)}\|_2$. We can solve and find the positive root, which gives
\begin{align*}
     \|\hat{v}^{(k)}\|_2^2&\leq 8h^q\tau^{2-q}\leq 8h^q (2\lam_k/\phi_0)^{2-q}&=Ch^q\left(\frac{\E[(y^{(0)}_1)^2]\log p}{n_0\phi_0^2}+\frac{\E[(y^{(k)}_1)^2]\log p}{n_k\phi_0^2}\right)^{1-q/2}\\
     &\leq  h^q \left(\frac{\log p}{n_0\wedge n_k}\right)^{1-q/2}.
\end{align*}
As a result,
\begin{align*}
    & \|\hat{v}^{(k)}\|_1\leq Ch^q \left(\frac{\E[(y^{(0)}_1)^2]\log p}{n_0\phi_0^2}+\frac{\E[(y^{(k)}_1)^2]\log p}{n_k\phi_0^2}\right)^{1/2-q/2}\\
     &\leq Ch^q \left(\frac{\log p}{n_0\wedge n_k}\right)^{1/2-q/2}\wedge h^2.
\end{align*}
The rest of the proof follows from the proof of Theorem \ref{thm1-l0}.

\end{proof}

\section{More results on simulation}
\label{sec-simu-app}
We report the numerical results on the estimated sparse indices (\ref{cond-agg2}). Specifically, we report $\widehat{C}$, which is the empirical probability of occurring $\{\widehat{R}^{(k)}~\text{is among the first }~|\mA|~\text{smallest},~k\in \mA\}$. If $\widehat{C}=1$, then $\max_{k\in \mA}\widehat{R}^{(k)}\leq \min_{k\in \mA^c}\widehat{R}^{(k)}$. Hence, $\widehat{C}$ close to 1 is favorable. We note that the case where $\mA=\emptyset$ or $\mA=\{1,\dots,K\}$ are trivial. Hence, we only consider other cases in Table \ref{table-rank}.

\begin{table}
\caption{\label{table-rank}The proportion of the occurring $\max_{k\in \mA}\widehat{R}^{(k)}\leq \max_{k\in\mA^c}\widehat{R}^{(k)}$ in all the settings considered in Section \ref{sec-simu}. }
\centering
\begin{tabular}{|c|c|c|c|c|c|c|c|}
\hline
\multirow{2}{*}{$h$} & \multirow{2}{*}{$|\mA|$} & \multicolumn{2}{c|}{Identity} &  \multicolumn{2}{c|}{Homogeneous} &  \multicolumn{2}{c|}{Heterogeneous} \\
\cline{3-8}
 & &(i) & (ii) &(i) & (ii) &(i) & (ii)\\
\hline
2 & 4 &1.00&1.00&0.99& 0.98& 0.88 & 1.00\\
 & 8  & 1.00&1.00&0.99& 0.98& 0.91 & 1.00\\
 & 12  & 1.00&1.00&1.00&0.99& 0.98 & 1.00\\
 & 16  & 1.00&1.00&1.00&0.99& 0.99 & 1.00\\
 \hline
 6 & 4 & 0.92&1.00&0.98&0.98& 0.94 & 1.00 \\
 & 8  &0.90&1.00&0.97&0.97& 0.95 & 1.00\\
 & 12 &0.95& 1.00&0.98&0.99& 0.96 & 1.00\\
 & 16 & 0.93&1.00&0.99&0.99& 1.00 &1.00\\
\hline
 12 & 4 &0.84&1.00&0.93&0.98& 0.50 & 1.00 \\
 & 8 & 0.87&1.00&0.99& 0.98&0.74 & 1.00\\
 & 12 & 0.88&1.00&0.94&0.98&0.94 & 1.00\\
 & 16 & 1.00&1.00&0.98&0.99&0.89 & 1.00\\
\hline
\end{tabular}
\end{table}

 \section{More results on data application}

We report the list  of genes analyzed in Section \ref{sec-data} and some basic summary statistics.

\begin{table}
\caption{\label{ap-tab2}The list of genes analyzed in Section \ref{sec-data} and the number of auxiliary samples (Aux. studies), the average primary sample sizes (Avg. pri. ss), and the average size of all auxiliary samples (Avg. aux. ss).}
\centering
{\small
\begin{tabular}{|c|c|c|c|c|c|}
  \hline
 & Target genes & Aux. studies & Avg. pri. ss & Avg. aux. ss \\ 
 \hline
1 & ENSG00000155304 & 47 & 177 & 14837 \\ 
  2 & ENSG00000154721 & 47 & 177 & 14837 \\ 
  3 & ENSG00000154734 & 47 & 177 & 14837 \\ 
  4 & ENSG00000171189 & 45 & 177 & 14146 \\ 
  5 & ENSG00000156299 & 47 & 177 & 14837 \\ 
  6 & ENSG00000159228 & 47 & 177 & 14837 \\ 
  7 & ENSG00000142197 & 47 & 177 & 14837 \\ 
  8 & ENSG00000159261 & 23 & 182 & 8094 \\ 
  9 & ENSG00000157557 & 47 & 177 & 14837 \\ 
  10 & ENSG00000182093 & 47 & 177 & 14837 \\ 
  11 & ENSG00000185437 & 47 & 177 & 14837 \\ 
  12 & ENSG00000183036 & 47 & 177 & 14837 \\ 
  13 & ENSG00000157601 & 47 & 177 & 14837 \\ 
  14 & ENSG00000160180 & 47 & 177 & 14837 \\ 
  15 & ENSG00000160181 & 23 & 237 & 6939 \\ 
  16 & ENSG00000160226 & 47 & 177 & 14837 \\ 
  17 & ENSG00000142185 & 46 & 177 & 14354 \\ 
  18 & ENSG00000197381 & 47 & 177 & 14837 \\ 
  19 & ENSG00000182871 & 47 & 177 & 14837 \\ 
  20 & ENSG00000142173 & 47 & 177 & 14837 \\ 
  21 & ENSG00000160299 & 47 & 177 & 14837 \\ 
  22 & ENSG00000160307 & 47 & 177 & 14837 \\ 
  23 & ENSG00000160182 & 20 & 237 & 5543 \\ 
  24 & ENSG00000157542 & 19 & 177 & 4009 \\ 
  25 & ENSG00000156284 & 14 & 237 & 4775 \\ 
   \hline
\end{tabular}
}
\end{table}
	
\begin{figure}[H]
 \centering
 \includegraphics[height=7cm, width=0.95\textwidth]{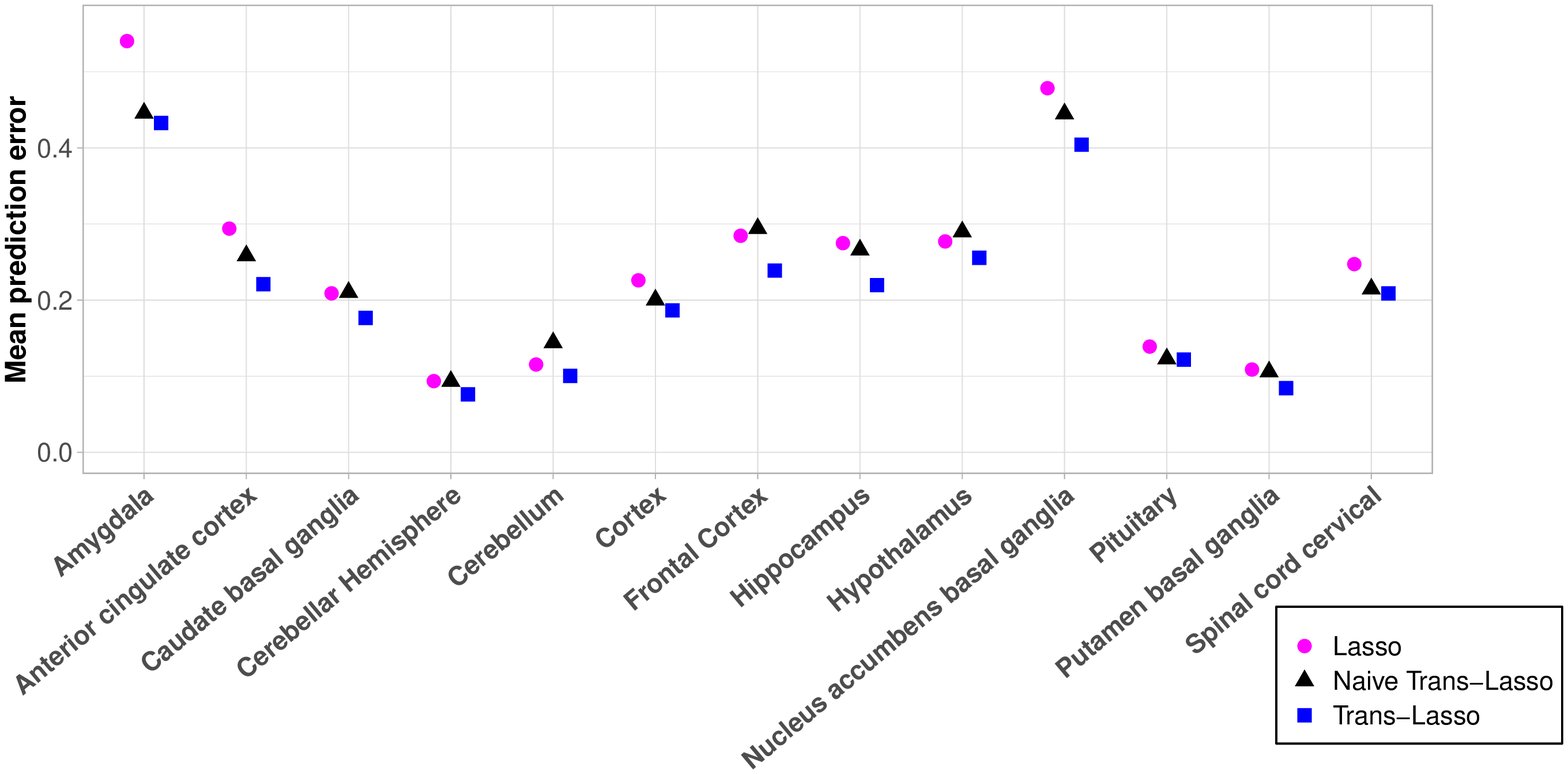}
 \caption{ \label{fig-data-3}Average prediction error of the Lasso, Naive-Trans-Lasso and Trans-Lasso via 5-fold cross validation for gene \texttt{JAM2} in multiple tissues.}
 \end{figure}

 \begin{figure}[H]
 \centering
 \makebox{\includegraphics[height=8cm, width=0.99\textwidth]{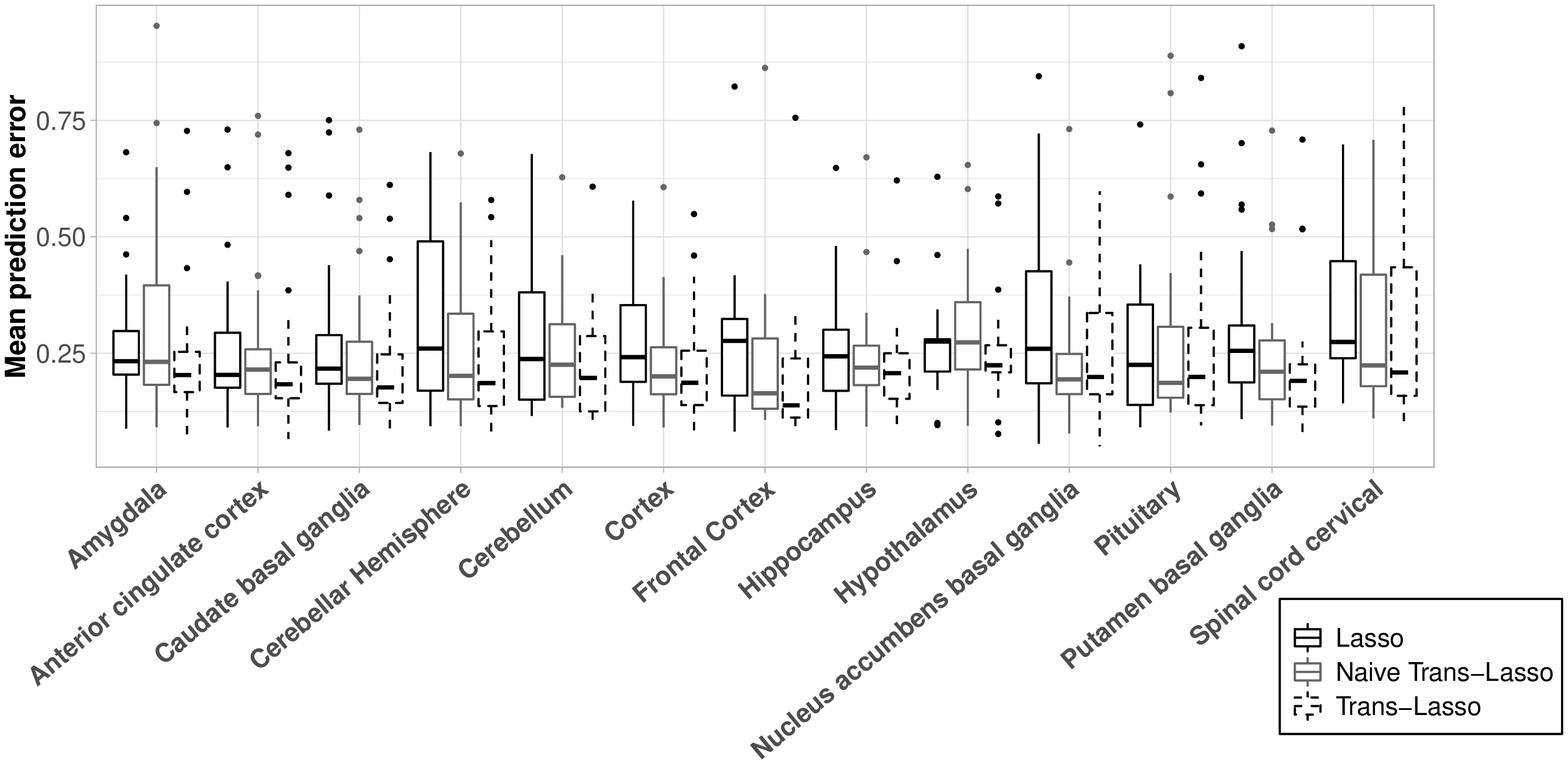}}
 \caption{ \label{fig-data-4}The overall prediction performance of the Lasso, Naive-Trans-Lasso, and Trans-Lasso for the 25 genes  on Chromosome 21 and in Module\underline{ }137, in multiple target tissues.}
 \end{figure}
\end{document}